\documentclass[runningheads]{llncs}

\newcommand{\sv}[1]{}
\newcommand{\lv}[1]{#1}

\lv{
\usepackage[letterpaper, margin=1in]{geometry}
}

\usepackage{times, microtype} 
\usepackage{url}
\usepackage[hidelinks]{hyperref}
\usepackage[utf8]{inputenc}
\usepackage[small]{caption}
\usepackage{graphicx}
\usepackage{amsmath}
\usepackage{booktabs}

\usepackage{amsmath,amssymb}
\usepackage{enumitem}

\usepackage{microtype}   

\usepackage[backgroundcolor=lightgray,colorinlistoftodos]{todonotes}

\presetkeys{todonotes}{inline,backgroundcolor=yellow}{}

\newcommand{\hy}{\hbox{-}\nobreak\hskip0pt}

\newcommand{\SB}{\{\,}%
\newcommand{\SM}{\;{:}\;}%
\newcommand{\SE}{\,\}}%
\newcommand{\Card}[1]{|#1|}

\newcommand{\R}{\mathbb{R}}
\newcommand{\Nat}{\mathbb{N}}

\newcommand{\NP}{\text{\normalfont NP}}
\newcommand{\FPT}{\text{\normalfont FPT}}

\newtheorem{THE}{Theorem} 
\newtheorem{LEM}[THE]{Lemma}
\newtheorem{OBS}[THE]{Observation} 
\newtheorem{PRO}[THE]{Proposition} 
\newtheorem{COR}[THE]{Corollary} 
 
\newtheorem{EXL}[THE]{Example} 
 
\newtheorem{DEF}[THE]{Definition} 

\newcommand{\BBB}{\mathcal{B}}
\newcommand{\JJJ}{\mathcal{J}}
\newcommand{\TTT}{\mathcal{T}}
\newcommand{\FFF}{\mathcal{F}}

\newcommand{\PPP}{\mathcal{P}}

\newcommand{\defproblem}[3]{
  \vspace{2pt}
\noindent\fbox{
  \begin{minipage}{0.96\columnwidth}
  #1 \\ 
  {\bf{Input:}} #2  \\
  {\bf{Question:}} #3
  \end{minipage}
  }
  \vspace{2pt}
}

\newcommand{\CSP}{\textsc{CSP}}
\newcommand{\cspi}{\mathbf{I}}
\newcommand{\tuple}[1]{\langle{#1}\rangle}  

\newcommand{\concat}{\circ}
\newcommand{\noftup}{\sharp_{\emph{tup}}}
\newcommand{\SOL}{\textup{SOL}}

\newcommand{\join}{\Join}
\newcommand{\proj}[2]{\pi_{#2}(#1)}

\newcommand{\bdbi}{\beta}
\newcommand{\jtbi}{\varrho}
\newcommand{\jtbie}{\rho}
  
\newcommand{\PR}{\textsf{prune}}

\newcommand{\jw}{\mathsf{jw}}
\newcommand{\jwl}{\mathsf{ljw}}

\newcommand{\CE}[1]{E}

\newcommand{\fecm}{\gamma}
\newcommand{\fecn}{\mathsf{fec}}
\newcommand{\fhtw}{\mathsf{fhtw}}
\newcommand{\guards}{\mathsf{\delta}}
\newcommand{\bw}{\mathsf{bw}}

\newcommand{\lbw}{\mathsf{lbw}}

\newcommand{\weight}{w}

\newcommand{\bigoh}{\mathcal{O}}

\begin{document}

  \title{A Join-Based Hybrid Parameter for \\ Constraint Satisfaction\thanks{Robert Ganian acknowledges support by the Austrian Science Fund (FWF, Project P31336) and is also affiliated with FI MUNI, Czech Republic.}}

%

\author{Robert Ganian\inst{1} \and
Sebastian Ordyniak\inst{2} \and Stefan Szeider\inst{1}}

\authorrunning{R.\ Ganian and S.\ Ordyniak and S.\ Szeider}
\institute{Algorithms and Complexity group, TU Wien, Vienna, Austria \and 
Algorithms group, University of Sheffield, Sheffield, UK}


\maketitle

\begin{abstract}
  We propose \emph{joinwidth}, a new complexity parameter for the
  Constraint Satisfaction Problem (CSP). The definition of
  joinwidth is based on the arrangement of basic operations on
  relations (joins, projections, and pruning), which
  inherently reflects the steps required to solve the instance.
  We use joinwidth to obtain polynomial-time algorithms (if a
  corresponding decomposition is provided in the input) as well as
  fixed-parameter algorithms (if no such decomposition is provided) for
  solving the CSP.
  
  Joinwidth is a \emph{hybrid} parameter, as it takes both the
  graphical structure as well as the constraint relations that appear
  in the instance into account. It has, therefore, the potential to
  capture larger classes of tractable instances than purely
  \emph{structural} parameters like hypertree width and the more
  general fractional hypertree width (fhtw). Indeed, we show that any
  class of instances of bounded fhtw also has bounded joinwidth, and
  that there exist classes of instances of bounded joinwidth and
  unbounded fhtw, so bounded joinwidth properly generalizes bounded
  fhtw.


  We further show that bounded joinwidth also properly generalizes
  several other known hybrid restrictions, such as fhtw with degree
  constraints and functional dependencies. In this sense, bounded
  joinwidth can be seen as a unifying principle that explains the
  tractability of several seemingly unrelated classes of CSP instances.
\end{abstract}
%
%
%

\section{Introduction}
The Constraint Satisfaction Problem (CSP) is a central and generic
computational problem that provides a common framework for many
theoretical and practical applications in AI and other areas of
Computer Science~\cite{RossiVanBeekWalsh06}. An
instance of the CSP consists of a collection of variables that must be
assigned values subject to constraints, where each constraint is given
in terms of a relation whose tuples specify the allowed combinations
of values for specified variables.

CSP is \NP-complete in general. A central line of research is
concerned with the identification of classes of instances for which
the CSP can be solved in polynomial time. The two main approaches are
to define classes either in terms of the constraint relations that may
occur in the instance (\emph{syntactic restrictions}; see, e.g.,
\cite{BulatovJeavonsKrokhin05}), or in terms of the constraint
hypergraph associated with the instance (\emph{structural
  restrictions}; see, e.g.,~\cite{Grohe07}). There are also several
prominent proposals for utilizing simultaneously syntactic and
structural restrictions called \emph{hybrid restrictions} (see, e.g.,
\cite{Naanaa13,CooperDMETZ16,CooperZ17,CohenCJZ19}).

Grohe and Marx~\cite{GroheMarx14} showed that CSP is polynomial-time
tractable whenever the constraint hypergraph has bounded
\emph{fractional hypertree width}, which strictly generalizes
previous tractability results based on hypertree width
\cite{GottlobLeoneScarcello02} and acyclic
queries~\cite{Yannakakis81}.  Bounded fractional hypertree
width is the most general known structural restriction 
that gives rise to polynomial-time tractability of CSP.

\lv{
As bounded fractional hypertree width is a structural restriction that
is completely oblivious to the relations present in the
instance, it is natural to expect that one can modify fractional
hypertree width to take the shape of relations into account. This is indeed
possible, but does not lead to a compact notion that is well-suited for
further theoretical analysis. 
}
\paragraph{Our contribution: Joinwidth.} 
We propose a new hybrid restriction for the CSP, the width parameter
\emph{joinwidth}, which is based on the arrangement of basic
relational operations along a tree, and not on hypertree
decompositions. Interestingly, as we will show, our notion strictly
generalizes (i)~bounded fractional hypertree width, (ii)~recently introduced
extensions of fractional hypertree width with degree constraints and
functional dependencies~\cite{Khamis0S17}, (iii)~various
prominent hybrid restrictions~\cite{CohenCGM11}, as well as 
(iv)~tractable classes  based on \emph{functionality} and \emph{root
  sets}~\cite{DevilleH91,David95,CohenCGM11}.
Hence, joinwidth gives rise to a common framework that captures several
different tractable classes considered in the past.
Moreover, none of the other hybrid parameters that we are aware of~\cite{CooperZ17},
such as classes based on the Broken Triangle Property or topological
minors~\cite{CooperDMETZ16,CohenCJZ19} and directional rank~\cite{Naanaa13}, 
generalize fractional hypertree width and hence all of them are either
less general or orthogonal to joinwidth.

Joinwidth is based on the arrangement of the constraints on the leaves
of a rooted binary tree which we call a \emph{join decomposition}. The
join decomposition indicates the order in which relational joins are
formed, where one proceeds in a bottom-up fashion from the leaves to
the root, labeling a node by the join of the relations at its
children, and projecting away variables that do not occur in relations
to be processed later. Join decompositions are related to (structural)
\emph{branch decompositions} of hypergraphs, where the hyperedges are
arranged on the leaves of the
tree~\cite{AlekhnovichRazborov02,Grohe07a,RobertsonSeymour83}. Related
notions have been considered in the context of query
optimization~\cite{AhmedSenPoessChakappen14,IoannidisKang91}. However,
the  basic form of  join decompositions using only relational joins and
projections is still a weak notion that cannot be used
to tackle instances of bounded fractional hypertree width efficiently.
We identify a further operation that---in conjunction with
relational joins and projections---gives rise to the powerful new
concept of joinwidth that captures and extends the various known
tractable classes mentioned.
This third operation \emph{prunes} away all the tuples from an
intermediate relation that are inconsistent with a relation to be
processed later. 


A join decomposition of a CSP instance specifies the order in which
the above three operations are applied, and its \emph{width} is the
smallest real number $w$ such that each relation appearing within the
join decomposition has at most $m^w$ many tuples (where $m$ is the
maximum number of tuples appearing in any constraint relation of the
CSP instance under consideration). 
The joinwidth of a CSP
instance is the smallest width over all its join decompositions.
Observe that joinwidth is a hybrid parameter---it depends on both the graphical
structure as well as the constraint relations appearing in the
instance.


 
\paragraph{Exploiting Joinwidth.}
Similarly to other width parameters, also the property that a class of
CSP instances has bounded joinwidth can only be exploited for CSP
solving if a decomposition (in our case a join decomposition)
witnessing the bounded width is provided as part of the input. While
such a join decomposition can be computed efficiently from a
fractional hypertree decomposition or when the CSP instance belongs to
a tractable class based on functionality or root sets mentioned
earlier, we show that computing an optimal join decomposition is
\NP-hard in general, mirroring the corresponding \NP-hardness of
computing optimal fractional hypertree
decompositions~\cite{FischlGottlobPichler18}. 
  
However, this obstacle disappears if we move from the viewpoint of polynomial-time tractability to
\emph{fixed-parameter tractability} (\FPT).
Under the FPT viewpoint, one considers classes of instances $\cspi$
that can be solved by a fixed\hy parameter algorithm---an algorithm
running in time $f(k)|\cspi|^{O(1)}$, where $k$ is the parameter
(typically the number of variables or constraints), $|\cspi|$ is
the 
size of the instance, and $f$ is a computable
function~\cite{FlumGrohe06,GottlobSzeider08,Grohe02}.  We note that it
is natural to assume that $k$ is much smaller than $\Card{\cspi}$ in
typical cases. The use of fixed-parameter tractability
is well motivated in the CSP setting; see, for instance, Marx's
discussion on this topic~\cite{Marx13}.

Here, we obtain two single-exponential fixed\hy parameter algorithms
for instances of bounded joinwidth (i.e., algorithms with a running time of
$2^{\bigoh(k)}\cdot |\cspi|^{\bigoh(1)}$): one where $k$ is the number
of variables, and the other when $k$ is the number of constraints.  In
this setting, we do not require an associated join decomposition to be
provided with the input.

Under the FPT viewpoint, Marx~\cite{Marx13} previously introduced the structural
parameter \emph{submodular width} (bounded submodular width
  is equivalent to bounded \emph{adaptive width} \cite{Marx11}),
which is strictly more general than fractional hypertree width, but
when bounded only gives rise to fixed\hy parameter tractability and not
polynomial-time tractability of CSP.  In fact, Marx showed that
assuming the Exponential Time Hypothesis~\cite{ImpagliazzoPaturiZane01}, bounded submodular width is the most
general purely structural restriction that yields fixed\hy parameter
tractability for CSP.  
However, as joinwidth is a hybrid parameter, it can (and we show that it does) remain bounded even on instances of unbounded submodular width---and the same holds also for the recently introduced extensions of submodular
width based on functional dependencies and degree bounds~\cite{Khamis0S17}.

\paragraph{Roadmap.}
After presenting the required
preliminaries on (hyper-)graphs, \CSP{}, and
fractional hypertree width in Section~\ref{sec:pre}, we introduce and
motivate join decompositions and joinwidth in
Section~\ref{sec:jw}. We establish some fundamental properties
of join decompositions, provide our tractability result for \CSP{} for the
case when a join decomposition is given as part of the input, and then obtain our \NP{}-hardness result for computing join
decompositions of constant width. Section~\ref{sec:justification} provides an in-depth
justification for the various design choices underlying join
decompositions; among others, we show that the pruning step is
required if the aim is to generalize fractional
hypertree width. Our algorithmic applications for joinwidth are
presented in Section~\ref{sec:tract-classes}: for instance, we show that
joinwidth generalizes fractional hypertree width, but also other known
(and hybrid) parameters such as functionality, root sets, and Turan
sets. Section~\ref{sec:solving} contains our fixed\hy parameter
tractability results for classes of CSP instances with bounded
joinwidth. Finally, in Section~\ref{sec:beyond}, we compare the algorithmic power
of joinwidth to the power of algorithms which rely on the unrestricted use of join and projection operations. 

\section{Preliminaries}\label{sec:pre}

We will use standard graph terminology~\cite{Diestel12}. 
An \emph{undirected graph} $G$ is a pair $(V,E)$, where $V$ or
$V(G)$ is the vertex set and $E$ or $E(G)$ is the edge set. 
All our graphs are simple and loopless. 
For a tree $T$ we use $L(T)$ to denote the set of its leaves.
For $i\in \Nat$, we let $[i]=\{1,\dots,i\}$.

\sv{\paragraph{Hypergraphs.}} 
\lv{\subsection{Hypergraphs, Branchwidth and Treewidth}}
%
Similarly to graphs, a \emph{hypergraph} $H$ is a pair $(V,E)$ where $V$ or
$V(H)$ is its vertex set and $E$ or $E(H)\subseteq 2^V$ is its set of hyperedges.
We denote by $H[V']$ the hypergraph \emph{induced} on the vertices in $V' \subseteq V$, i.e.,
the hypergraph with vertex set $V'$ and edge set $\SB e\cap V' \SM e \in E \SE$.
\lv{
Every subset $F$ of $E(H)$ defines a
\emph{cut} of $H$, i.e., the pair $(F,E(H)\setminus F)$. We denote by
$\guards_H(F)$ (or just $\guards(F)$ if $H$ is clear from the context)
the set of \emph{cut vertices} of $F$ in $H$, i.e., $\guards(F)$
contains all vertices incident to both an edge in $F$ and an edge in
$E(H) \setminus F$. Note that $\guards(F)=\guards(E(H)\setminus F)$.

Let $H$ be a hypergraph. A \emph{branch-decomposition} of $H$ is a
pair $\BBB=(B,\bdbi)$, where $B$ is a rooted binary tree and $\bdbi : L(B)
\rightarrow E(H)$ is a bijection between the leaves $L(B)$ of $B$ and
the edges of $H$. For simplicity, we write $\bdbi(B')$ to denote
the set $\SB \bdbi(l) \SM l \in L(B')\SE$ of edges for a subtree $B'$ of $B$.
The \emph{branchwidth} of an edge $e$ of $B$, denoted by $\bw(e)$,
is equal to $|\guards_H(\bdbi(B'))|$, where $B'$ is any of the two
components of $B - e$. The \emph{branchwidth} of $\BBB$ is equal to
$\max_{e \in E(B)}\bw(e)$ and the branchwidth of $H$, denoted by $\bw(H)$, is the minimum
branchwidth of any of its branch-decompositions. We say that $\BBB$ is
a \emph{linear branch decomposition} if every inner node of $B$ is
adjacent to at least one leaf node and define the \emph{linear
  branchwidth} of $H$, denoted by $\lbw(H)$, as the minimum
branchwidth over all linear branch decompositions of $H$.

A \emph{tree-decomposition} of $H$ is a pair $\TTT=(T,(B_t)_{t \in
  V(T)})$, where $T$ is a tree and $B_t \subseteq V(H)$ for every $t
\in V(T)$ such that: (1) for every $e \in E(H)$ there is a $t \in
V(T)$ such that $e \subseteq B_t$ and (2) for every $v \in V(H)$ the
set $\SB t \in V(T) \SM v \in B_t \SE$ induces a non-empty subtree of $T$. The
\emph{treewidth} of $\TTT$ is equal to
$\max_{t \in V(T)}(|B_t|-1)$, and the \emph{treewidth} of $H$ is the
minimum treewidth of any tree-decomposition of $H$. We say that $\TTT$
is a \emph{path-decomposition} if $T$ is a path and define the
\emph{pathwidth} of $H$ as the minimum treewidth
of any path-decomposition of $H$. 
}
\smallskip

\sv{\paragraph{The Constraint Satisfaction Problem.}}
\lv{\subsection{The Constraint Satisfaction Problem}}
Let $D$ be a set and $n$ and $n'$ be natural numbers. 
An $n$-ary relation on $D$
is a subset of $D^n$. For a tuple $t \in D^n$, we denote by $t[i]$,
the $i$-th entry of~$t$, where $1 \leq i \leq n$. For
two tuples $t \in D^n$ and $t' \in D^{n'}$, we denote by $t \concat
t'$, the concatenation of $t$ and $t'$.

An instance of a \emph{constraint satisfaction problem} (CSP) $\cspi$
is a triple $\tuple{V,D,C}$, where $V$ is a finite set of variables
over a finite set (domain) $D$, and $C$ is a set of constraints. A
\emph{constraint} $c \in C$ consists of a \emph{scope},
denoted by $S(c)$, which is a completely ordered subset of
$V$, and a relation, denoted by $R(c)$, which is a
$|S(c)|$-ary relation on $D$. 
We let $|c|$ denote the number of tuples in $R(c)$ and $|\cspi|=|V|+|D|+\sum_{c\in C}|c|$. Without loss of generality, we assume that each variable occurs in the scope of at least one constraint.

A \emph{solution} for $\cspi$ is an assignment $\theta : V \rightarrow D$ of the variables in $V$ to
domain values (from $D$) such that for every constraint $c \in C$ with scope $S(c)=(v_1,\dotsc,v_{|S(c)|})$,
the relation $R$ contains the tuple $\theta(S(c))=(\theta(v_1),\dotsc,\theta(v_{|S(c)|}))$. 
We denote by $\SOL(\cspi)$ the constraint
containing all solutions of $\cspi$, i.e., the constraint with scope
$V=\{v_1,\dotsc,v_n\}$, whose relation contains one tuple
$(\theta(v_1),\dotsc,\theta(v_n))$ for every solution $\theta$ of
$\cspi$. The task in \CSP{} is to decide whether the instance $\cspi$
has at least one solution or in other words whether $\SOL(\cspi)\neq
\emptyset$. Here and in the following we will for convenience (and
with a slight abuse of notation) sometimes treat constraints like sets of tuples.

For a variable $v \in S(c)$ and a tuple $t \in R(c)$, we denote
by $t[v]$, the $i$-th entry of $t$, where $i$ is the position of $v$ in
$S(c)$. Let $V'$ be a subset of $V$ and let $V''$ be all the variables that appear in $V'$ and $S(c)$.
With a slight abuse of notation, we denote by $S(c)\cap V'$, the sequence $S(c)$
restricted to the variables in $V'$ and we denote by $t[V']$ the tuple $(t[v_1],\dotsc,t[v_{|V''|}])$,
where $S(c) \cap V'=(v_1,\dotsc,v_{|V''|})$.

Let $c$ and $c'$ be two constraints of $\cspi$.
We denote by $S(c) \cup S(c')$, the ordered set (i.e., tuple) $S(c) \concat (S(c')
\setminus S(c))$. The
\emph{(natural) join} between $c$ and $c'$, denoted by $c \join c'$, is
the constraint with scope
$S(c) \cup S(c')$ containing all tuples $t \concat t'[S(c')\setminus S(c)]$ such that $t \in
R(c)$, $t' \in R(c')$, and $t[S(c) \cap S(c')]=t'[S(c) \cap S(c')]$.
The \emph{projection} of $c$ to $V'$, denoted by $\proj{c}{V'}$, is the constraint
with scope $S(c)\cap V'$, whose relation contains all tuples $t[V']$
with $t \in R(c)$. We note that if $c$ contains at least one tuple, then projecting it onto a set $V'$ with $V'\cap S(c)=\emptyset$ results in the constraint with an empty scope and a relation containing the empty tuple (i.e., a tautological constraint). On the other hand, if $R(c)$ is the relation containing the empty tuple, then every projection of $c$ will also result in a relation containing the empty tuple.

For a CSP instance $\cspi=\tuple{V,D,C}$ we sometimes denote by $V(\cspi)$,
$D(\cspi)$, $C(\cspi)$, and $\noftup(\cspi)$ its set of variables $V$, its domain $D$,
its set of constraints $C$, and the maximum number of tuples in any
constraint relation of $\cspi$, respectively. For a subset $V'
\subseteq V$, we will also use $\cspi[V']$ to denote the sub-instance of $\cspi$ induced by the
variables in $V'\subseteq V$, i.e., $\cspi[V']=\tuple{V',D,\SB \proj{c}{V'}
\SM c \in C \SE}$. 
The \emph{hypergraph} $H(\cspi)$ of a CSP instance $\cspi=\tuple{V,D,C}$ is the
hypergraph with vertex set $V$ and edge set $\SB S(c) \SM c \in C\SE$.

It is well known that for every instance $\cspi$ and every instance
$\cspi'$ obtained by either (1) replacing two constraints in
$C(\cspi)$ by their natural join or (2) adding a projection of a
constraint in $C(\cspi)$, it holds that $\SOL(\cspi)=\SOL(\cspi')$. As
a consequence, $\SOL(\cspi)$ can be computed by performing, e.g., a
sequence of joins over all the constraints in~$C$.

\newcommand{\projC}{\PPP}



\smallskip
\sv{\paragraph{Fractional Hypertree Width.}}
\lv{\subsection{Fractional Hypertree Width}}
Let $H$ be a hypergraph.
A \emph{fractional edge cover} for $H$ is a mapping $\fecm : E(H)
\rightarrow \R$ such that $\sum_{e \in E(H) \land v \in e}\fecm(e)\geq
1$ for every $v \in V(H)$. The \emph{weight} of $\fecm$, denoted by
$\weight(\fecm)$, is the number $\sum_{e \in E(H)}\fecm(e)$. The
\emph{fractional edge cover number} of $H$, denoted by $\fecn(H)$, is
the smallest weight of any fractional edge cover of~$H$.

A \emph{fractional hypertree decomposition} $\TTT$ of $H$ is a triple
$\TTT=(T,(B_t)_{t \in V(T)},$ $(\fecm_t)_{t \in V(T)})$, where $(T,(B_t)_{t
  \in V(T)})$ is a tree decomposition~\cite{RobertsonS86,DowneyFellows13} of $H$ and $(\fecm_t)_{t \in
  V(T)}$ is a family of mappings from $E(H)$ to $\R$ such that for
every $t \in V(T)$, it holds that $\fecm_t$ is a fractional edge cover
for $H[B_t]$.
We call the sets $B_t$ the \emph{bags} and
the mappings $\fecm_t$ the \emph{fractional guards} of the
decomposition.
The \emph{width} of $\TTT$ is the maximum $\weight(\gamma_t)$ over all
$t \in V(T)$. The \emph{fractional hypertree width} of $H$, denoted by $\fhtw(H)$, is the
minimum width of any fractional hypertree decomposition of
$H$. Finally, the \emph{fractional hypertree width} of a CSP instance $\cspi$,
denoted by $\fhtw(\cspi)$, is equal to $\fhtw(H(\cspi))$.

\begin{PRO}
  \label{pro:fhtw-bag-sol}
  \sloppypar
  Let $\cspi$ be a CSP instance with hypergraph $H$ and let
  $\TTT=(T,(B_t)_{t \in V(T)},(\fecm_t)_{t \in V(T)})$ be a 
  fractional hypertree decomposition of $H$ of width at most $\omega$. For every node $t\in V(T)$ 
  and every subset $B \subseteq B_t$, it holds that
  $|\SOL(\cspi[B])|\leq (\noftup(\cspi))^{\omega}$.
\end{PRO}
\lv{
\begin{proof}
 It follows from the definition of fractional hypertree width that
 $\fecm_t$ is a fractional edge cover for $H[B_t]$ of width at most
 $\omega$ for every $t \in V(T)$. Since fractional edge covers are heriditary, $\fecm_t$ is
 also a fractional edge cover for $H[B]$ of width at most $\omega$
 for every $B \subseteq B_t$. Moreover, it is shown in~\cite[Lemma
 3]{GroheMarx14}, that any CSP instance $\cspi'$ has at most
 $(\noftup(\cspi'))^{\fecn(H(\cspi'))}$ solutions;
 note that~\cite[Lemma 3]{GroheMarx14} actually only states that
 $\cspi'$ has at most $(|\cspi'|)^{\fecn(H(\cspi'))}$ solutions,
 however, the slightly stronger bound of
 $(\noftup(\cspi'))^{\fecn(H(\cspi'))}$ follows immediately from the
 proof of~\cite[Lemma 3]{GroheMarx14}. Hence, because
 $\fecn(\cspi[B])\leq \omega$, we obtain that $\cspi[B]$ has at most
 $(\noftup(\cspi))^{\omega}$ solutions.
\qed \end{proof}
}

\section{Join Decompositions and Joinwidth}
\label{sec:jw}

This section introduces two notions that are central to our contribution: \emph{join decompositions} and \emph{joinwidth}. 
In the following, let us consider an arbitrary CSP instance $\cspi=\tuple{V,D,C}$.

\begin{DEF}\label{def:jointree}
  A \emph{join decomposition} for $\cspi$ is a pair $(J,\jtbi)$, where $J$ is a rooted binary
  tree and $\jtbi$ is a bijection between the leaves $L(J)$ of $J$ and $C$.
\end{DEF}
Let $j$ be a node of $J$. We denote by $J_j$ the subtree of $J$ rooted at $j$
and we denote by $X(j)$, $V(j)$, $\overline{V}(j)$, and $S(j)$ the (unordered) sets
$\SB \jtbi(\ell) \SM \ell\in L(J_j)\SE$, $\bigcup_{c \in
  X(j)}S(c)$, $\bigcup_{c \not \in X(j)}S(c)$, and $V(j)\cap\overline{V}(j)$, respectively; infuitively, $X(j)$ is the set of constraints that occur in the subtree rooted at $j$, $V(j)$ is the set of variables that occur in the scope of constraints in $X(j)$, $\overline{V}(j)$ is the set of variables that occur in the scope of constraints not in $X(j)$, and $S(j)$ is the set of variables that occur in $V(j)$ and $\overline{V}(j)$.
  In some cases, we will also consider \emph{linear join decompositions}, which are join decompositions 
where every inner node is
adjacent to at least one leaf. 

\paragraph{Semantics of Join Decompositions.}
Intuitively, every internal node of a join decomposition represents a join operation that is carried out over the constraints obtained for the two children; in this way, a join decomposition can be seen as a procedure for performing joins, with the aim of determining whether $\SOL(\cspi)$ is non-empty (i.e., solving the CSP instance~$\cspi$). Crucially, the running time of such a procedure depends on the size of the constraints obtained and stored by the algorithm which performs such joins. The aim of this subsection is to formally define and substantiate an algorithmic procedure which uses join decompositions to solve CSP.

\lv{
\paragraph{The Naive Approach.} }
A naive way of implementing the above idea would be to simply compute
and store the natural join at each node of the join decomposition and
proceed up to the root; see for instance the work of \cite{AtseriasGroheMarx13}. Formally, we can recursively define a constraint
$C_{\emph{naive}}(j)$ for every node $j \in V(J)$ as follows. 
If $j$ is a leaf, then
$C_{\emph{naive}}(j)=\jtbi(j)$. Otherwise $C_\emph{naive}(j)$ is equal to
$C_{\emph{naive}}(j_1) \join C_{\emph{naive}}(j_2)$, where $j_1$ and $j_2$ are the two
children of $j$ in $J$. 
It is easy to see that this approach can create large constraints even for very simple instances of CSP: for example, at the root $r$ of $T$ it holds that $\SOL(\cspi)=C_{\emph{naive}}(r)$, and hence $C_{\emph{naive}}(r)$ would have superpolynomial size for every instance of CSP with a superpolynomial number of solutions. In particular, an algorithm which computes and stores $C_{\emph{naive}}(j)$ would never run in polynomial time for CSP instances with a superpolynomial number of solutions.

\lv{\paragraph{A Better Approach} 
A more effective way}
\sv{An efficient way} of joining constraints along a join
decomposition is to only store projections of constraints onto those variables that are still relevant for constraints which have yet to appear; this idea has been used, e.g., in algorithms which exploit hypertree width~\cite{GottlobLeoneScarcello02}. To formalize this, let $C_{\emph{proj}}(j)$ be recursively defined for every node $j \in V(T)$ as follows. 
If $j$ is a leaf, then
$C_{\emph{proj}}(j)=\proj{\jtbi(j)}{\overline{V}(j)}$. Otherwise $C_\emph{proj}(j)$ is equal to
$\big(\proj{C_{\emph{proj}}(j_1) \join C_{\emph{proj}}(j_2)\big)}{\overline{V}(j)}$, where $j_1$ and $j_2$ are the two children of $j$ in $J$. In this case, $\cspi$ is a YES-instance if and only if $C_{\emph{proj}}(r)$ does not contain the empty relation. Clearly, for every node $j$ of $J$ it holds that $C_{\emph{proj}}(j)$ has at most as many tuples as $C_{\emph{naive}}(j)$, but can have arbitrarily fewer tuples; in particular, an algorithm which uses join decompositions to compute $C_{\emph{proj}}$ in a bottom-up fashion can solve CSP instances in polynomial time even if they have a superpolynomial number of solutions (see also Observation~\ref{obs:compnaive}).

However, the above approach still does not capture the algorithmic power offered by dynamically computing joins along a join decomposition. In particular, 
similarly as has been done in the evaluation algorithm 
for fractional edge cover~\cite[Theorem
3.5]{GroheMarx14}, we can further reduce the size of each constraint $C_{\emph{proj}}(j)$ computed in the above procedure by \emph{pruning} all tuples that would immediately violate a constraint $c$ in $\cspi$ (and, in particular, in $C\setminus C(j)$). 
To formalize this operation, we let $\PR(c)$ denote the \emph{pruned constraint} w.r.t.\ $\cspi$,
i.e., $\PR(c)$ is obtained from $c$ by 
removing all tuples $t\in R(c)$ such that there is a constraint $c' \in C$ with
$t[S(c')] \notin \proj{c'}{S(c)}$.
This leads us to our final notion of dynamically computed constraints: for a node $j$, we let $C(j)=\PR(C_{\emph{proj}}(j))$. 
We note that this, perhaps inconspicuous, notion of pruning is in fact critical---without it, one cannot use join decompositions to efficiently solve instances of small fractional hypertree width or even small fractional edge cover. A more in-depth discussion on this topic is provided in Section~\ref{sec:justification}.

We can now proceed to formally define the considered width measures. 
\begin{DEF}\label{def:width}
Let $\JJJ=(J,\jtbi)$ be a join decomposition for $\cspi$ and let $j \in V(J)$.
The \emph{joinwidth} of $j$, denoted $\jw(j)$, 
is the smallest real number
$\omega$ such that $|C(j)| \leq (\noftup(\cspi))^\omega$, i.e., $\omega=\log_{\noftup(\cspi)}|C(j)|$.
The joinwidth of $\JJJ$ (denoted $\jw(\JJJ)$) is then the maximum
$\jw(j)$ over all $j \in V(J)$. Finally, the joinwidth of $\cspi$
(denoted $\jw(\cspi)$) is the minimum $\jw(\JJJ)$ over all join decompositions
$\JJJ$ for $\cspi$.
\end{DEF}
  In general terms, an instance $\cspi$ has joinwidth $\omega$ if it admits a join decomposition where the number of tuples of the produced constraints never increases beyond the $\omega$-th power of the size of the largest relation in $\cspi$. Analogously as above, we denote by $\jwl(\cspi)$ the minimum joinwidth of any linear join decomposition of a CSP instance $\cspi$. 

\begin{EXL}\label{ex:join}
  Let $N\in \mathbb{N}$ and consider the CSP instance $\cspi$ having three variables $a$, $b$,
  and $c$ and three constraints $x$, $y$, and $z$ with scopes $(a,b)$,
  $(b,c)$, and $(a,c)$, respectively. Assume furthermore that the relations of all
  three constraints are identical and contain all tuples $(1,i)$ and
  $(i,1)$ for every $i \in [N]$. Refer also to
  Figure~\ref{fig:join-example} for an illustration of the example.
  Then $|x|=|y|=|z|=\noftup(\cspi)=2N-1$
  and due to the symmetry of $\cspi$ any join-tree $\JJJ$ of $\cspi$
  has the same joinwidth, which (as we will show) is equal to $1$.
  To see this consider for instance the join-tree $\JJJ$ that
  has one inner node $j$ joining $x$ and $y$ and a root node $r$
  joining $C(j)$ and $z$. Then $\jw(\ell)=1$ for any leaf node $\ell$ of
  $\JJJ$. Moreover $|C(j)|=|\PR(C_{\emph{proj}}(j))|=|z|=\noftup(\cspi)$ since the pruning
  step removes all tuples from $C_{\emph{proj}}(j)$ that are not in
  $z$ and consequently $C(r)=z$ and $\jw(\JJJ)=1$. Note that in this example $\jw(\cspi)=1< \fhtw(\cspi)=3/2$.
\end{EXL}

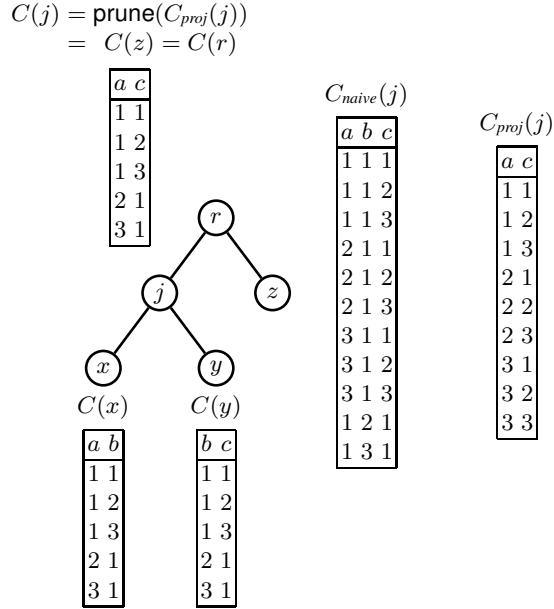
\begin{figure}
  \centering
  \vspace{-0.5cm}
  \begin{tikzpicture}[line width=2pt, node distance=0.75cm]
    \tikzstyle{jtn}=[draw, shape=circle, line width=1pt, minimum size=13pt,inner sep=1pt,align=center]
    \tikzstyle{jne}=[draw, line width=1pt]

    \draw
    node[jtn] (x) {$x$}
    node[right of=x] (xym) {}
    node[jtn, right of=xym] (y) {$y$}

    node[jtn, node distance=1cm, above of=xym] (j) {$j$} 
    node[right of=j] (xyzm) {}
    node[jtn, right of=xyzm] (z) {$z$}
    node[jtn, node distance=1cm, above of=xyzm] (r) {$r$} 
    ;
    
    \draw
    (x) edge[jne] (j)
    (y) edge[jne] (j)

    (j) edge[jne] (r)
    (z) edge[jne] (r)
    ;
    \draw[label distance=-4pt]
    (r) +(-1.5cm,-4cm) node[label=above:$C(x)$] (Rx) {
     \small \begin{tabular}{|cc|}
        \hline
        $a$ & $b$ \\\hline
        $1$ & $1$ \\
        $1$ & $2$ \\
        $1$ & $3$ \\
        $2$ & $1$ \\
        $3$ & $1$ \\\hline
      \end{tabular}
    }
    (r) +(-0cm,-4cm) node[label=above:$C(y)$] (Ry) {
      \small\begin{tabular}{|cc|}
        \hline
        $b$ & $c$ \\\hline
        $1$ & $1$ \\
        $1$ & $2$ \\
        $1$ & $3$ \\
        $2$ & $1$ \\
        $3$ & $1$ \\\hline
      \end{tabular}
    }

    (r) +(2cm,-1cm) node[label=above:$C_{\emph{naive}}(j)$] (RJxy) {
     \small \begin{tabular}{|ccc|}
        \hline
        $a$ & $b$ & $c$\\\hline
        $1$ & $1$ & $1$ \\
        $1$ & $1$ & $2$ \\
        $1$ & $1$ & $3$ \\
        $2$ & $1$ & $1$ \\
        $2$ & $1$ & $2$ \\
        $2$ & $1$ & $3$ \\
        $3$ & $1$ & $1$ \\
        $3$ & $1$ & $2$ \\
        $3$ & $1$ & $3$ \\
        $1$ & $2$ & $1$ \\
        $1$ & $3$ & $1$ \\
        \hline
      \end{tabular}
    }

    (r) +(4cm,-1cm) node[label=above:$C_{\emph{proj}}(j)$] (RPJxy) {
      \small\begin{tabular}{|cc|}
        \hline
        $a$ & $c$\\\hline
        $1$ & $1$ \\
        $1$ & $2$ \\
        $1$ & $3$ \\
        $2$ & $1$ \\
        $2$ & $2$ \\
        $2$ & $3$ \\
        $3$ & $1$ \\
        $3$ & $2$ \\
        $3$ & $3$ \\
        \hline
      \end{tabular}
    }

    (r) +(-1.15cm,0.8cm)
    node[label=above:{ $
      \begin{array}{c@{\;=\;}c}
C(j)&\PR(C_{\emph{proj}}(j))\\&C(z)=C(r)        
      \end{array}$}] (Rxy) 
    {
      \small\begin{tabular}{|cc|}
        \hline
        $a$ & $c$\\\hline
        $1$ & $1$ \\
        $1$ & $2$ \\
        $1$ & $3$ \\
        $2$ & $1$ \\
        $3$ & $1$ \\
        \hline
      \end{tabular}
    }

    ;
  \end{tikzpicture}
  \caption{The join decomposition given in Example~\ref{ex:join} for $N=3$ together with the intermediate constraints
    obtained for the node $j$.
    \vspace{-0.5cm}}
  \label{fig:join-example}
\end{figure}

Finally, we remark that one could in principle also define joinwidth in terms of a (rather tedious and technically involved) variant of hypertree decompositions. However, the inherent algorithmic nature of join-trees makes them much better suited for the definition of joinwidth.

\paragraph{Properties of Join Decompositions.}
%
%
Our first task is to formalize the intuition behind the
constraints $C(j)$ computed when proceeding through the join
tree. 

\begin{LEM}
\label{lem:jtp-part-sol}
  Let $(J,\jtbi)$ be a join decomposition for $\cspi=\tuple{V,D,C}$ and let $j \in
  V(J)$. Then $C(j)=\proj{\SOL(\cspi')}{\overline{V}(j)}$, where $\cspi'=\cspi[V(j)]$.
\end{LEM}
\begin{proof}
  We prove the lemma by leaf-to-root induction along $J$. If $j$ is a leaf such that $\jtbi(j)=c$, then $C(j)$ is the constraint obtained from $c$ by projecting onto $\overline{V}(j)$ and then applying pruning with respect to $\cspi$. 
  Crucially, pruning $c$ w.r.t.\ $\cspi'$ produces the same result as pruning $c$ w.r.t. $\cspi$.
Since pruning cannot remove tuples which occur in $\SOL(\cspi')$, each tuple in $\SOL(\cspi')$ must also occur in $C(j)$ (as a projection onto $\overline{V}(j)$). On the other hand, consider a tuple $\alpha$ in $C(j)$ and assume for a contradiction that $\alpha$ is not present in $\proj{\SOL(\cspi')}{\overline{V}(j)}$. Since variables outside of $\overline{V}(j)$ do not occur in the scopes of constraints other than $c$, this means that there would exist a constraint $c'$ in $\cspi'$ which is not satisfied by an assignment corresponding to $\alpha$---but in that case $\alpha$ would be removed from $C(j)$ via pruning. Hence $C(j)=\proj{\SOL(\cspi')}{\overline{V}(j)}$ holds for every leaf in $T$.

For the induction step, consider a node $j$ with children $j_1$ and $j_2$ (with their corresponding instances being $\cspi'_1$ and $\cspi'_2$, respectively), and recall that $C(j)$ is obtained from $C(j_1)\join C(j_2)$ by projecting onto $\overline V(j)$ and then pruning (w.r.t.\ $\cspi$ or, equivalently, w.r.t.\ $\cspi'$). We will also implicitly use the fact that $\overline V(j)\subseteq \overline V(j_1)\cup \overline V(j_2)$ and $V(j)= V(j_1) \cup V(j_2)$. First, consider for a contradiction that there exists a tuple $\beta$ in $\SOL(\cspi'[\overline V(j)])$ which does not occur in $C(j)$. Clearly, $\beta$ could not have been removed by pruning, and hence this would mean that there exists no tuple in $C(j_1)\join C(j_2)$ which results in $\beta$ after projection onto $\overline V(j)$; in particular, w.l.o.g.\ we may assume that every tuple in $C(j_1)$ differs from $\beta$ in (the assignment of) at least one variable. However, since $\beta$ occurs in $\SOL(\cspi'[\overline V(j)])$, there must exist at least one tuple, say $\beta'$, which occurs in $\SOL(\cspi')$, and consequently there exists a tuple in $\SOL(\cspi_1')$ which matches $\beta$ in (the assignment of) all variables. At this point, we have reached a contradiction with the inductive assumption that $\SOL(\cspi_1'[\overline V(j_1)])=C(j_1)$.

For the final case, consider a tuple $\gamma$ in $C(j)$ and assume for a contradiction that $\gamma$ is not present in $\proj{\SOL(\cspi')}{\overline{V}(j)}$. This means that there exists at least one constraint, say $c'$, in $\proj{\SOL(\cspi')}{\overline{V}(j)}$ which would be invalidated by (an assignment corresponding to) $\gamma$. Let us assume that $c'$ occurs in the subtree rooted in $j_2$, and let $\gamma_1$ be an arbitrary ``projection'' of $\gamma$ onto $\overline V(j_1)$. Since $\overline V(j_1)\supseteq S(c')$, this means that $\gamma_1$ would have been removed from $C(j_1)$ by pruning; in particular, we see that there exists no tuple $\gamma_1$ in $C(j_1)$ which could produce $\gamma$ in a join, contradicting our assumptions about $\gamma$. By putting everything together, we conclude that indeed $C(j)=\proj{\SOL(\cspi')}{\overline{V}(j)}$.
\qed \end{proof}

Next, we show how join decompositions can be used to solve \CSP{}.
\begin{THE}
\label{the:jt}
  \CSP{} can be solved in time $\bigoh(|\cspi|^{2\omega+4})$ provided that a
  join decomposition of width at most $\omega$ is given in the input.
\end{THE}
\begin{proof}
  Let $\JJJ=(J,\jtbi)$ be the provided join decomposition of width $\omega$ for
  $\cspi$. As noted before, the algorithm for solving $\cspi$ computes
  $C(j)$ for every $j \in V(J)$ in a bottom-up manner. Since $J$ has
  exactly $2|C|-1$ nodes, it remains to analyse the maximum time
  required to compute $C(j)$ for any node of $J$. If $j$ is a leaf,
  then $C(j)=\PR(\proj{\jtbi(j)}{S(j)})$ and since the time required to
  compute the projection $P=\proj{\jtbi(j)}{S(j)}$ from $\jtbi(j)$ is at most $\bigoh(\noftup(\cspi)|S(j)|)$
  and the time required to compute the pruned constraint
  $\PR(P)$ from $P$ is at most $\bigoh(|C|(\noftup(\cspi))^2|S(j)|)$,
  we obtain that $C(j)$ can be computed in time
  $\bigoh(|C|(\noftup(\cspi))^2|S(j)|)\in \bigoh(|\cspi|^4)$.
  Moreover, if $j$ is an inner node with children $j_1$ and $j_2$, then
  $C(j)=\PR(\proj{C(j_1)\join C(j_2)}{S(j)})$ and since we require
  at most $\bigoh((\noftup(\cspi))^{2\omega}|S(j_1)\cup S(j_2)|)$ time
  to compute the join $Q=C(j_1)\join C(j_2)$ from $C(j_1)$
  and $C(j_2)$, at most $\bigoh((\noftup(\cspi))^{2\omega}|S(j)|)$
  time to compute the projection $P=\proj{Q}{S(j)}$ from $Q$, and at
  most $\bigoh(|C|(\noftup(\cspi))^{2\omega+1}\cdot |S(j)|)$ time to compute the
  pruned constraint $\PR(P)$ from $P$, we obtain
  $\bigoh(|C|(\noftup(\cspi))^{2\omega+1}\cdot |S(j_1)\cup S(j_2)|)=\bigoh(|\cspi|^{2\omega+3})$ as the
  total time required to compute $C(j)$. Multiplying the time required
  to compute $C(j)$ for an inner node $j\in V(J)$ with the number of
  nodes of $T$ yields the running time stated in the lemma.
\qed \end{proof}

\paragraph{Computing Join Decompositions.}
Next, let us address the problem of
computing join decompositions of bounded joinwidth, formalized as follows.
\newcommand{\COMPJT}{\textsc{Join Decomposition}}

\defproblem{$\omega$-\COMPJT}{A CSP instance $\cspi$.}{Compute a join decomposition for $\cspi$ of width at most $\omega$, or correctly determine that $\jw(\cspi)>\omega$.}

%
%
We show that $\omega$-\COMPJT\ is \NP{}-hard even for width $\omega=1$. This is similar to fractional hypertree width, where it was only
very recently shown that deciding whether $\fhtw(\cspi)\leq 2$ is
\NP{}-hard~\cite{FischlGottlobPichler18}, settling a question which had been
open for about a decade.
Our proof is, however, entirely different from
the corresponding hardness proof for fractional hypertree width and
uses a reduction from the \NP{}-complete \textsc{Branchwidth} problem~\cite{SeymourT94}.


\begin{THE}
\label{the:comp-jt-np}
  $1$-\COMPJT{} is \NP-hard, even on Boolean CSP instances.
\end{THE}
\lv{
\begin{proof}
  We will reduce from the well-established \NP{}-complete \textsc{Branchwidth} problem, which given a
  graph $G$ and an integer $\omega$ asks, whether $G$ has branchwidth
  at most $\omega$. Let $(G,\omega)$ be an instance of the
  \textsc{Branchwidth} problem. We will first show how to construct a
  ternary CSP instance
  $\cspi=(V,D,C)$ such that $\bw(G)\leq \omega$ if and only if
  $\jw(\cspi)\leq \log_{|V(G)|+2}|V(G)|+\omega\leq 2$; we
  will then later explain how to adapt the construction of $\cspi$ to
  make its domain boolean and to reduce the required joinwidth
  to $1$. Let $V(G)=\{v_1,\dotsc,v_n\}$. We set:
  \begin{itemize}
  \item $V=V(G) \cup \{a\}$,
  \item $D=\{1,\cdots,n\}$,
  \item for every $e=\{v_i,v_j\} \in E(G)$, $\cspi$ has a constraint
    $c_e$ with scope $(a,v_i,v_j)$ whose relation $R(c_e)$
    contains every tuple $t$ such that:
    \begin{itemize}
    \item $t[a] \in \{1,\cdots,n\}$,
    \item for $l \in \{i,j\}$, it holds that $t[v_l]\in \{1,2\}$
      if $t[a]=l$ and $t[v_l]=1$ otherwise.
    \end{itemize}
  \end{itemize}
  Note that $\noftup(\cspi)=n+2$. Since $\cspi$ can clearly be
  constructed in polynomial time it only
  remains to show that $\bw(G)\leq \omega$ if and only if
  $\jw(\cspi)\leq \log_{n+2}(n+\omega)\leq 2$.

  
  Towards showing the forward direction, let $\BBB=(B,\bdbi)$ be a branch
  decomposition for $G$ of width at most $\omega$. We claim that
  $\JJJ=(B,\jtbi)$, where $\jtbi(l)=c_{\bdbi(l)}$ for every $l \in
  L(B)$, is a join decomposition for $\cspi$ of width at most $\log_{n+2}(n+\omega)$. Consider any
  node $j\in V(B)$, then $S(j)=\guards(B_j)\cup \{a\}$, moreover, the
  constraint $C(j)$ contains every tuple $t$ such that:
  \begin{itemize}
  \item $t[a] \in \{1,\cdots,n\}$,
  \item for every $v_i \in S(j)\setminus \{a\}$, it holds that $t[v_i]\in \{1,2\}$
    if $t[a]=i$ and $t[v_i]=1$ otherwise.
  \end{itemize}
  Hence
  $\jw(j)=\log_{\noftup(\cspi)}|C(j)|=\log_{n+2}(n+|\guards(j)|)\leq
  \log_{n+2}(n+\omega)$. Since this holds for every $j \in V(J)$, we obtain
  that $\jw(\JJJ)=\log_{n+2}(n+\omega)$, as required.
  
  Towards showing the reverse direction, let $\JJJ=(J,\jtbi)$ be a join decomposition
  for $\cspi$ of width at most $\log_{n+2}(n+\omega)$. We claim that
  $\BBB=(J,\bdbi)$, where $\bdbi(l)=e$ if $\jtbi(l)=c_e$ for every $l \in
  L(B)$, is a branch decomposition for $G$ of width at most $\omega$. Consider any
  node $j\in V(J)$, then $\guards(B_j)=S(j)\setminus \{a\}$, moreover, the
  constraint $C(j)$ contains every tuple $t$ such that:
  \begin{itemize}
  \item $t[a] \in \{1,\cdots,n\}$,
  \item for every $v_i \in S(j)\setminus \{a\}$, it holds that $t[v_i]\in \{1,2\}$
    if $t[a]=i$ and $t[v_i]=1$ otherwise.
  \end{itemize}
  \sloppypar Hence $|C(j)|=n+|\guards(j)|$, which because $|C(j)|\leq
  (\noftup(\cspi))^{\log_{\noftup(\cspi)}(n+\omega)}=n+\omega$ implies
  that $|\guards(j)|\leq \omega$. Since this holds for every node $j
  \in V(J)$, we obtain that $\bw(G)\leq \omega$, as required.
  The completes the first part of the proof. We will now show how to
  adapt the construction of $\cspi$ in such a way that $\bw(G)\leq
  \omega$ if and only $\jw(\cspi)\leq 1$. The idea is to artificially
  increase $\noftup(\cspi)$ from $n+2$ to $n+\omega$ without changing
  the properties of $\cspi$ too much. To achieve this we simply
  introduce a new variable $b$ with domain $\{1,\dotsc,n+\omega\}$ and
  add the complete constraint $c$ with scope $(b)$. It is easy to see
  that the reduction still works for this slightly modified CSP
  instance and moreover we obtain that $\bw(G)\leq \omega$ if and only
  if $\jw(\JJJ)=\log_{n+\omega}n+\omega=1$.

  Finally, it is easy to convert $\cspi$ into a boolean CSP instance
  by replacing the variables $a$ and $b$ with $\log n$ respectively
  $\log n+\omega$ boolean variables.
\qed \end{proof}
}



%

\lv{\paragraph{Partial Join Decompositions.} Our final task in this section is to introduce 
the concept of partial join decompositions, which will
serve as a useful tool in later sections. Let $\cspi=(V,D,C)$ be a CSP
instance and let $\JJJ=(J,\jtbi)$ be a join decomposition for $\cspi$.
We say that $\JJJ'=(J_j,\jtbi_j)$ is a \emph{partial join decomposition} of
$\JJJ$ if $j\in V(J)$ and $\jtbi_j$ is the
restriction of $\jtbi$ to $L(J_j)$. Note that the semantics of $\JJJ'$, i.e., the notions
$C(j)$ as well as $\jw(j)$ (for every node $j\in V(J_j)$, are independent of the concrete
join decomposition $\JJJ$, but only depend on $\JJJ'$ and $\cspi$. We
therefore say that $\JJJ'$ is a \emph{partial join decomposition} for $\cspi$
covering the constraints in $C' \subseteq C$ if $\jtbi(L(J_j))=C'$;
assuming that the semantics for every node $j \in V(J_j)$ are defined
as if $\JJJ'$ would be part of an arbitrary join decomposition for~$\cspi$.
}

\section{Justifying Joinwidth}
\label{sec:justification}
Below, we substantiate the use of both pruning and projections in our
definition of join decomposition. In particular, we show that using pruning
and projections allows the joinwidth to be significantly lower than
if we were to consider joins carried out via $C_{\emph{naive}}$ or
$C_{\emph{proj}}$. More importantly, we show that join decompositions without
pruning do not cover CSP instances with bounded fractional edge cover
(and by extension bounded fractional hypertree width).  To formalize this, let
$\jw_{\emph{naive}}(\cspi)$ and $\jw_{\emph{proj}}(\cspi)$ be defined
analogously as $\jw(\cspi)$, with the distinction being that these 
measure the width in terms of $C_{\emph{naive}}$ and
$C_{\emph{proj}}$ instead of $C$.

We also justify the use of trees for join decompositions by showing that
there is an arbitrary difference between linear join decompositions (which precisely correspond to simple sequences of joins) and
join decompositions.


\begin{OBS} 
\label{obs:compnaive}
  For every integer $\omega$ there exists a CSP instance $\cspi_{\omega}$ such
  that \hbox{$\jw_{\emph{proj}}(\cspi_\omega)\leq 1$}, but $\jw_{\emph{naive}}(\cspi_\omega)\geq \omega$.
\end{OBS}
\begin{proof}
Consider the CSP instance $\cspi_\omega$ with variables
$x,v_1,\dots,v_\omega$ and for each $i\in [\omega]$ a constraint $c_i$
with scope $\{x,v_i\}$ containing the tuples $\tuple{0,1}$ and
$\tuple{0,0}$. Since $\SOL(\cspi_\omega)$ contains $2^\omega$ tuples,
it follows that every join decomposition with root $r$ must have
$|C(r)_{\emph{naive}}| = (\noftup(\cspi))^{\omega}=2^\omega$ tuples,
hence $\jw_{\emph{naive}}(\cspi_\omega)\geq \omega$.

On the other hand, consider a linear join decomposition which
introduces the constraints in an arbitrary order. Then for each inner
node $j$, it holds that $S(j)=\{x\}$ and in particular
$C_{\emph{proj}}(j)$ contains a single tuple $(0)$ over scope
$\{x\}$. We conclude that $\jw_{\emph{proj}}(\cspi_\omega)\leq 1$.
\qed \end{proof}

\lv{
Towards showing that linear joinwidth differs from joinwidth in an arbitrary manner, we will use the following lemma, which
establishes a tight relationship between the (linear) joinwidth of a
CSP instance and the branchwidth of its hypergraph when all
constraints of the CSP instance are \emph{complete}, i.e., their
relations contain all possible tuples.
\begin{LEM}
\label{lem:jw-bw}
  Let $\cspi=\tuple{V,D,C}$ be CSP instance that has only complete
  constraints. Then $\jw(\cspi)=\bw(H(\cspi))/2$ and
  $\jwl(\cspi)=\lbw(H(\cspi))$.
\end{LEM}
\begin{proof}
  Let $\cspi=\tuple{V,D,C}$ be the given CSP instance with hypergraph
  $H=H(\cspi)$. First note that any join decomposition $\JJJ=(J,\jtbi)$ for
  $\cspi$ is also a branch decomposition for $H$, whose width is equal
  to $\max_{j \in V(J)}|S(j)|$. Conversely, any branch decomposition
  of $H$ is also a join decomposition for $\cspi$. Moreover, because all
  constraints of $\cspi$ are complete, we have that $C(j)$ is equal to the
  complete constraint on $|S(j)|$ variables for every node $j$ in a
  join decomposition $\JJJ=(J,\jtbi)$ for $\cspi$. Hence,
  \begin{eqnarray*}
    \jw(\cspi) & = & \max_{j \in V(J)}\log_{\noftup(\cspi}|C(j)| \\
               & = &\max_{j \in V(J)}\log_{d^2}(d^{|S(j)|})\\
               & = & \max_{j \in V(J)}|S(j)|/2=\bw(H)/2
  \end{eqnarray*}
  The proof for
  the linear versions of join decompositions and branch decompositions is
  analogous.
\qed \end{proof}

Using the above Lemma we are now ready to establish an arbitrary
difference between join decompositions and linear join decompositions.}

\sv{The next proposition justifies the use of trees instead of
  just linear join decompositions. Its proof employs an interesting
  connection between branchwidth and joinwidth.}
\begin{PRO}
\label{pro:comptree}
  For every integer $\omega$ there exists a CSP instance $\cspi_{\omega}$ such
  that $\jw(\cspi_\omega)\leq 1$ but $\jwl(\cspi_\omega)\geq \omega$.
\end{PRO}
\lv{
\begin{proof}
  It is well known that trees have branchwidth one but can have
  arbitrarily high pathwidth~\cite{Diestel95}. Since pathwidth is
  known to be within a constant factor of linear
  branchwidth~\cite{Nordstrand17}, we obtain that for every $\omega$,
  there is a tree $T_\omega$ with $\bw(T_\omega)\leq 1$ but
  $\lbw(T_\omega)\geq \omega$. Now let $\cspi_{\omega}$ be the (boolean) CSP instance
  containing only full binary constraints such that
  $H(\cspi_\omega)=T_{2\omega}$. Then because of Lemma~\ref{lem:jw-bw} we have
  that $\jw(\cspi_\omega)=\bw(T_{2\omega})/2=0.5\leq 1$ and
  $\jwl(\cspi_\omega)=\lbw(T_{2\omega})/2=\omega$, as
  required.
\qed \end{proof}
}

\lv{
The next proposition is by far the most technically challenging of the three, and requires notions which will only be introduced later on. For this reason, we postpone its proof to the end of Subsection~\ref{sub:fract}. The proposition shows not only that pruning can significantly reduce the size of stored constraints, but also that without pruning (i.e., with projections alone) one cannot hope to generalize structural parameters such as fractional hypertree width.}
\sv{
The next proposition shows not only that pruning can significantly reduce the size of stored constraints, but also that without pruning (i.e., with projections alone) one cannot hope to generalize structural parameters such as fractional hypertree width.}

\begin{PRO}
\label{pro:compproj}
  For every integer $\omega$ there exists a CSP instance
  $\cspi_{\omega}$ with hypergraph $H^\omega$ such that
  $\jw(\cspi_\omega)\leq 2$ and $\fecn(H^\omega)\leq 2$ (and hence also $\fhtw(H^\omega)\leq 2$), but
  $\jw_{\emph{proj}}(\cspi_\omega)\geq \omega$.
\end{PRO}

%

We believe that the above results are of general interest, as they
provide useful insights into how to best utilize the joining of
constraints.

\section{Tractable Classes}\label{sec:tract-classes}
Here, we show that join decompositions of small width not only allow us to solve a wide range of CSP instances, but also provide a unifying reason for the tractability of previously established structural parameters and tractable classes.

\subsection{Fractional Hypertree Width}
\label{sub:fract}

We begin by showing that joinwidth is a strictly more general parameter than
fractional hypertree width. We start with a simple example showing that
the joinwidth of a CSP instance can be arbitrarily smaller than its
fractional hypertree width. Indeed, this holds for any
structural parameter $\psi$ measured purely on the hypergraph representation, i.e., we say that $\psi$ is a \emph{structural
parameter} if $\psi(\cspi)=\psi(H(\cspi))$ for any CSP instance~$\cspi$. Examples for structural parameters include fractional
and generalized hypertree width, but also \emph{submodular width}~\cite{Marx13}.
\begin{OBS}
\label{obs:fhtw}
  Let $\psi$ be any structural parameter such that for every $\omega$ there is a CSP
  instance with $\psi(\cspi)=\psi(H(\cspi))\geq \omega$. Then for every $\omega$ there is a CSP instance $\cspi_\omega$ with
  $\jw(\cspi_\omega)\leq 1$ but $\psi(\cspi_\omega)\geq \omega$.
\end{OBS}
\lv{
\begin{proof}
  Let $H_\omega$ be any hypergraph with $\psi(H_\omega)\geq \omega$. Then the CSP
  instance $\cspi_\omega$ obtained from $H_\omega$ by replacing every hyperedge with an
  empty constraint satisfies $\psi(\cspi_\omega)\geq \omega$ and
  $\jw(\cspi_\omega)\leq 1$. 
\qed \end{proof}
}
\lv{Note that it is straightforward to construct more interesting examples
that also show an arbitrary difference between joinwidth and the
recently introduced extensions of fractional hypertree width (or even
submodular width) with degree constraints and functional
dependencies~\cite{Khamis0S17}. For instance, when comparing joinwidth
with fractional hypertree width, there are many possibilities
to fill the constrains in such a way that
the join-width remains low and every such possibility leads to a new
example showing the difference between these two width measures. To
ensure this it would, e.g., be sufficient to ensure that every set of
constrains that appear together in a bag of a fractional hypertree
decomposition have small join-width; one such example is given in
Proposition~\ref{pro:compproj}.}
The following theorem shows that, for the case of fractional
hypertree width, the opposite of the above observation
is not true.

\begin{THE}\label{the:jtw-fhtw}
  For every CSP instance $\cspi$, it holds that $\jw(\cspi) \leq \fhtw(\cspi)$. 
\end{THE}

\begin{proof}
  Let $H$ be the hypergraph of the given CSP instance $\cspi=(V,D,C)$ and
  let $\TTT=(T, (B_t)_{t \in V(T)}, (\gamma_t)_{t \in V(T)})$ be an
  optimal fractional hypertree decomposition of $H$. We prove the
  theorem by constructing a join decomposition $\JJJ=(J,\jtbi)$ for $\cspi$,
  whose width is at most $\fhtw(H)$.
  Let $\alpha : E(H) \rightarrow V(T)$ be some function from the edges
  of $H$ to the nodes of $T$ such that $e \subseteq B_{\alpha(e)}$ for
  every $e \in E(H)$. Note that such a function always exists, because
  $(T,(B_t)_{t \in V(T)})$ is a tree decomposition of $H$. We denote
  by $\alpha^{-1}(t)$ the set $\SB e \in E(H) \SM \alpha(e)=t \SE$.

  The construction of $\JJJ$ now
  proceeds in two steps. First we construct a partial join decomposition
  $\JJJ^t=(J^t,\jtbi^t)$ for $\cspi$ that covers only the constraints
  in $\alpha^{-1}(t)$, for every $t \in V(T)$. 
  Second, we show 
  how to combine all the partial join decompositions into the join decomposition $\JJJ$ for $\cspi$ of
  width at most $\fhtw(H)$. 

  Let $t \in V(T)$ and let $\JJJ^t=(J^t,\jtbi^t)$ be an arbitrary partial join decomposition
  for $\cspi$ that covers the constraints in $\alpha^{-1}(t)$. Let us
  consider an arbitrary node $j \in V(J^t)$.
  By Lemma~\ref{lem:jtp-part-sol}, we know that
  $C(j)=\proj{\SOL(\cspi[V(j)])}{S(j)}$. Moreover, the fact that $\bigcup_{e \in \alpha^{-1}(t)}e
  \subseteq B_t$ implies $V(j) \subseteq B_t$.
  Since $|\proj{\SOL(\cspi[V(j)])}{S(j)}|\leq |\SOL(\cspi[V(j)])|$, by invoking
  Proposition~\ref{pro:fhtw-bag-sol} we obtain that
  $|C(j)|\leq |\SOL(\cspi[V(j)])| \leq
  \noftup(\cspi)^{\fhtw(\cspi)}$. Hence we conclude that $\jw(j)\leq
  \fhtw(H)$. 
  
  Next, we show how to combine the partial join decompositions $\JJJ^t$ into the
  join decomposition $\JJJ$ for $\cspi$. We will do this via a bottom-up
  algorithm that computes a (combined) partial join decomposition
  $\FFF^t=(F^t,\jtbie^t)$ (for every node $t \in V(T)$) that covers all constraints in
  $\alpha^{-1}(T_t)=\bigcup_{t \in V(T)}\alpha^{-1}(t)$.
  Initially, we set
  $\FFF^l=\JJJ^l$ for every leaf $l \in L(T)$. 
  For a
  non-leaf $t \in V(T)$ with children $t_1,\dotsc,t_\ell$ in $T$, we
  obtain $\FFF^t$ from the already computed partial join decompositions
  $\FFF^{t_1},\dotsc,\FFF^{t_\ell}$ as follows. Let $P$ be a path on the
  new vertices $p_1,\dotsc,p_\ell$ and let $r_t$ and
  $r_{t_1},\dotsc,r_{t_\ell}$ be the root nodes of $J^t$ and
  $F^{t_1},\dotsc,F^{t_\ell}$, respectively. Then we obtain $F^t$ from the disjoint
  union of $P$, $J^t$, $F^{t_1},\dotsc,F^{t_\ell}$ after adding an edge
  between $r_t$ and $p_1$ and an edge between $r_{t_i}$ and $p_i$ for
  every $i$ with $1\leq i \leq \ell$ and setting $p_\ell$ to be the root of
  $F^t$. Moreover, $\jtbie^t$ is obtained as the combination (i.e., union) of the
  functions $\jtbi^t$, $\jtbie^{t_1},\dotsc,\jtbie^{t_\ell}$. 
  Observe that because $\alpha$ assigns every hyperedge to 
  precisely one bag of $T$, it holds that every constraint assigned to $T_t$
  is mapped to precisely one leaf of $\FFF^t$. At this point, all that
  remains is to show that $\FFF^t$ has joinwidth at most $\fhtw(H)$.

  Since we have already argued that $|\SOL(\cspi[V(j)])| \leq
  \noftup(\cspi)^{\fhtw(H)}$ for every node $j$ of $\JJJ^t$ and
  moreover we can assume that the same holds for every node $j$ of
  $\FFF^{t_1},\dotsc,\FFF^{t_\ell}$ by the induction hypothesis, it only remains to
  show that the same holds for the nodes $p_1,\dotsc,p_\ell$. First, observe
  that since $(T,(B_t)_{t \in V(T)})$ is a tree decomposition of
  $H$, it holds that $S(r_t),S(r_{t_1}),\dotsc,S(r_{t_\ell}) \subseteq
  B_t$. Indeed, consider for a contradiction that, w.l.o.g., there exists a variable $x\in S(r_{t_1})\setminus B_t$. Then there must exist a hyperedge $e_1\ni x$ mapped to $t_1$ or one of its descendants, and another hyperedge $e_i\ni x$ mapped to some node $t'$ that is neither $t_1$ nor one of its descendants. But then both $B_{t'}$ and $B_{t_1}$ must contain $x$, and so $B_t$ must contain $x$ as well. Moreover, since $S(p_i)\subseteq (S(r_t)\cup S(r_{t_1})\cup \dots \cup S(r_{t_\ell}))$ for every $i\in [\ell]$, it follows that $S(p_i)\subseteq B_t$ as well.
  
  Finally, recall that $C(p_i)=\proj{\SOL(\cspi[V(p_i)])}{S(p_i)}$ by
  Lemma~\ref{lem:jtp-part-sol}, and observe that
  $|\proj{\SOL(\cspi[V(p_i)])}{S(p_i)}|\leq
  |\SOL(\cspi[S(p_i)])|$. Then by Proposition~\ref{pro:fhtw-bag-sol}
  combined with the fact that $S(p_i)\subseteq B_t$, we obtain
  $|C(p_i)|\leq |\SOL(\cspi[S(p_i)])|\leq \noftup(\cspi)^{\fhtw(H)}$,
  which implies that the width of $p_i$ is indeed at most $\fhtw(H)$.
\qed \end{proof}

\lv{
Since we have now established the relationship between joinwidth and fractional hypertree width, we can conclude this subsection with a proof of Proposition~\ref{pro:compproj} (cf.\ Section~\ref{sec:justification}).

\begin{proof}[Proof of Proposition~\ref{pro:compproj}]
  The proof uses an adaptation of a construction given by Atserias, Grohe and Marx~\cite[Theorem 7]{AtseriasGroheMarx13}. Let $m=4\omega+1$ and $n=\binom{2m}{m}$. The CSP instance
  $\cspi_\omega$ has:
  \begin{itemize}
  \item one variable $v_S$ for every $S \subseteq [2m]$ with $|S|=m$,
  \item one constraint $c_i$ for every $i$ with $1\leq i \leq 2m$ with
    scope $\SB v_S \SM i \in S \SE$. The relation of $c_i$ contains
    the following tuples $t$ for every $v \in S(c_i)$: every tuple $t$
    such that $t[v] \in [n]$ and $t[u]=1$ for every $u\neq v$,
  \item domain $[n]$.
  \end{itemize}
  Let $H$ be the hypergraph of $\cspi_\omega$. Note first that $H$ has $n$
  vertices, $2m$ edges, every edge of $H$ has size
  $\binom{2m-1}{m-1}=n/2$, every vertex of $H$ occurs in exactly $m$
  edges, and every set of at most $m$ edges of $H$ misses at least one
  vertex, i.e., the union of at most $m$ edges of $H$ misses at least one vertex from $V(H)$.
  Note furthermore that $\noftup(\cspi_\omega)=\frac{n^2}{2}$.

  It is shown in~\cite[Theorem 4]{AtseriasGroheMarx13} that
  $H$ has a fractional edge cover of size at most $2$. This
  is witnessed by the mapping $\fecm : E(H) \rightarrow \R$ defined by
  setting $\fecm(e)=1/m$ for every $e\in E(H)$. Because every variable $v_S$
  of $\cspi_\omega$ is in the scope of exactly $m$ constraints (i.e. the
  constraints in $\SB c_i \SM i\in S \SE$), $\fecm$ is indeed a fractional
  edge cover and because $\sum_{e \in E(H)}\fecm(e)=2m(1/m)=2$ it has
  size exactly $2$. Note that because $\fecn(H)\geq \fhtw(H)$, we
  obtain from Theorem~\ref{the:jtw-fhtw} that
  $jw(\cspi_\omega)\leq 2$.

  It remains to show that $\jw_{\emph{proj}}(\cspi_\omega)\geq \omega$. Consider
  an arbitrary join decomposition $\JJJ=(J,\jtbi)$ of $\cspi_\omega$. Let $j$
  be a node of $J$ such that $\lceil m/2 \rceil \leq |L(J_j)|< m$. Note
  that such a node $j$ always exists due to the following well-known
  combinatorial argument.

  First we show that there is a node $j'$ of $J$ such that
  $|L(J_{j'})| \geq m$
  but $|L(J_{j''})| < m$
  for every child $j''$ of $j'$ in $J$. Namely, $j'$ can be found by
  going down the tree $J$ starting from the root and choosing a child
  $j''$ of $j'$ with $|L(J_{j''})| \geq m$, as long as
  such a child $j''$ exists. Then letting $j$ be the child
  of $j'$ maximizing $|L(J_j)|$ implies that
  $\lceil m/2 \rceil \leq |L(J_j)| < m$,
  as required.

  We claim that $\jw_{\emph{proj}}(j) \geq \omega$. First note that
  $S(j)=V(j)$. This is because $|C(\cspi_\omega)\setminus X(j)|\geq
  2m-(m-1)=m+1$ and hence $\overline{V}(j)=\bigcup_{c \in
    (C(\cspi_\omega)\setminus X(j))}S(c)=V(\cspi_\omega)$; this is because $S \cap \SB i \SM
  c_i \in C(\cspi_\omega)\setminus X(j)\SE\neq \emptyset$ for every
  variable $v_S$ of $\cspi_\omega$. Consequently, $C_{\emph{proj}}(j)$ is equal
  to the join of all constraints in $X(j)$. We show next that there
  are at least $\lceil m/2 \rceil$ variables that appear in the scope of exactly one
  constraint in $X(j)$, which implies that $C_{\emph{proj}}(j)$ contains every of
  the $n^{\lceil m/2 \rceil}$ tuples on these variables. To see this let $O$ be
  the set of all elements $o \in [2m]$ for which $c_o \in X(j)$ and
  let $\overline{O}$ be the set of all the remaining elements in $[2m]$. Then $|O|
  \geq \lceil m/2\rceil$ and $|\overline{O}|\geq 2m-(m-1)=m+1$. Let $S'$ be any set of
  exactly $m-1$ elements from $\overline{O}$ and for every $o \in O$, let $S_o$ be equal to
  $\{o\} \cup S'$. Then $v_{S_o}$ is in the scope of exactly one
  constraint in $C(j)$ (i.e., the constraint $c_o$ and hence there are
  at least $\lceil m/2\rceil$ variables that occur in the scope of exactly one
  constraint in $X(j)$, i.e., the variables $\SB S_o \SM o \in O \SE$.
  Hence $C_{\emph{proj}}(j)$ contains at least $n^{\lceil m/2\rceil}$ tuples, which
  because $\noftup(\cspi_\omega)=n^2/2$ implies that $\jw_{\emph{proj}}(j)$
  is at least $\lceil m/2\rceil/2\geq m/4\geq \omega$, as required.
\qed \end{proof}
}

%
\subsection{Functionality and Root Sets}
Consider a CSP instance $\cspi=\tuple{V,D,C}$ with $n=|V|$. We say
that a constraint $c\in C$ is \emph{functional} on variable $v\in V$
if $c$ does not contain two tuples that differ \emph{only} at variable
$v$; more formally, for every $t$ and $t'\in R(c)$ it holds that if
$t[v]\neq t'[v]$, then there exists a variable $z\in S(c)$ distinct
from $v$ such that $t[z]\neq t'[z]$. The instance $\cspi$ is then
called \emph{functional} if there exists a variable ordering $v_1<
\dots < v_n$ such that, for each $i\in [n]$, there exists a constraint
$c\in C$ such that $\proj{c}{\{v_1,\dots,v_i\}}$ is functional on
$v_i$.
Observe that every CSP instance that is functional can admit at most
$1$ solution~\cite{CohenCGM11}; this restriction can be relaxed through the notion of
\emph{root sets}, which can be seen as variable sets that form
``exceptions'' to functionality. Formally, a variable set $Q$ is a
root set if there exists a variable ordering $v_1< \dots < v_n$ such
that, for each $i\in [n]$ where $v_i\not \in Q$, there exists a
constraint $c\in C$ such that $\proj{c}{\{v_1,\dots,v_i\}}$ is
functional on $v_i$; we say that $Q$ is \emph{witnessed} by the
variable order $v_1< \dots < v_n$.

Functionality and root sets were studied for Boolean CSP
\cite{DevilleH91,David95}. Cohen et al.~\cite{CohenCGM11} later
extended these notions to the CSP with larger domains. Our aim in this
section is twofold:
\sv{
  (1) generalize root sets through the introduction of
  \emph{constraint root sets} and (2) show that bounded-size
  constraint root sets (and also root sets) form a special case of
  bounded joinwidth.
}
\lv{\begin{itemize}
\item generalize root sets through the introduction of \emph{constraint root sets};
\item show that bounded-size constraint root sets (and also root sets) form a special case of bounded joinwidth.
\end{itemize}
}
Before we proceed, it will be useful to show that one can always assume the root set to occur at the beginning of the variable ordering.

\begin{OBS}
Let $Q$ be a root set in $\cspi$ witnessed by a variable order $\alpha$, assume a fixed arbitrary ordering on $Q$, and let the set $V'=V(\cspi)\setminus Q$ be ordered based on the placement of its variables in $\alpha$. Then $Q$ is also witnessed by the variable order $\alpha'=Q\concat V'$.
\end{OBS}

\lv{
\begin{proof}
We need to show that $Q$ is a root set witnessed by $\alpha'=v_1<\dots<v_{|V|}$, i.e., that for each $i$ in $k<i\leq |V|$ there exists a constraint $c\in C$ such that $\proj{c}{\{v_1,\dots,v_i\}}$ is functional on $v_i$. Consider an arbitrary such variable $v_i$, and recall that since $Q$ is witnessed by $\alpha$, there must exist a constraint $c$ such that $c'=\proj{c}{Q_i\cup V_i}$ is functional on $v_i$, where $V'_i$ and $Q_i$ are, respectively, the elements of $Q$ and $V'$ which occur before $v_i$ in $\alpha$. Now, observe that the constraint $c^*=\proj{c}{Q\cup V_i}=\proj{c}{\{v_1,\dots,v_i\}}$ must also be functional on $v_i$---indeed, for each pair of tuples $t,t'$ in $R(c^*)$, either $t[v_i]=t'[v_i]$, or $t[v_i]\neq t'[v_i]$ and there exists a variable $z\in Q_i\cup V_i$ (and hence also in the scope of $c^*$) such that $t[z]\neq t'[z]$.
\qed \end{proof}
}

For ease of presentation, we will say that $\cspi$ is
$k$-\emph{rooted} if $k$ is the minimum integer such that $\cspi$ has
a root set of size $k$. It is easy to see, and also follows from the
work of David~\cite{David95} and Cohen et al.~\cite{CohenCGM11}, that
for every fixed $k$ the class of $k$-rooted CSP instances is polynomial-time
solvable: generally speaking, one can first loop through and test all
variable-subsets of size at most $k$ to find a root $Q$, and then loop
through all assignments $Q\rightarrow D$ to get a set of functional
CSP instances, each of which can be solved separately in linear time.

While even $1$-rooted CSP instance can have unbounded fractional
hypertree width (see also the discussion of Cohen et
al.~\cite{CohenCGM11}), the class of $k$-rooted CSP instances for a fixed value $k$ is, in some sense, not very robust. Indeed, consider the CSP instance $\mathcal{W}=\tuple{\{v_1,\dots,v_n\},\{0,1\},\{c\}}$ where $c$ ensures that precisely a single variable is set to $1$ (i.e., its relation can be seen as an $n\times n$ identity matrix). In spite of its triviality, it is easy to verify that $\mathcal{W}$ is not $k$-rooted for any $k<n-2$.

Let us now consider the following alternative to measuring the size of root sets in a CSP instance $\cspi$. A constraint set $P$ is a \emph{constraint-root set} if $\bigcup_{c\in P}S(c)$ is a root set, and $\cspi$ is then called $k$-constraint-rooted if $k$ is the minimum integer such that $\cspi$ has a constraint-root set of size $k$. Since we can assume that each variable occurs in at least one constraint, every $k$-rooted CSP also has a constraint-root set of size at most $k$; on the other hand, the aforementioned example of $\mathcal{W}$ shows that an instance can be $1$-constraint-rooted while not being $k$-rooted for any small $k$. The following result, which we prove by using join decompositions and joinwidth, thus gives rise to strictly larger tractable classes than those obtained via root sets:

\begin{PRO}
\label{pro:rootset}
For every fixed $k\in \mathbb{N}$, every $k$-constraint-rooted
CSP instance has joinwidth at most $k$ and can be solved in time $|\cspi|^{\bigoh(k)}$.
\end{PRO}

\begin{proof}
Consider a CSP instance $\cspi$ with a constraint-root set $P$ of size
$k$. We argue that $\cspi$ has a linear join decomposition of width at most $k$
where the elements of $P$ occur as the leaves farthest from the
root. Indeed, consider the linear join decomposition $(J,\jtbi)$ constructed in
a bottom-up manner, as follows. First, we start by gradually adding
the constraints in $P$ as the initial leaves. At each step after that,
consider a node $j$ which is the top-most constructed node in the
join decomposition. By definition, there must exist a variable $v$ and a
constraint $c$ such that $\proj{c}{\bigcup_{c\in P} S(c)}$ is
functional on $v$. Moreover, this implies that $|\proj{c}{\bigcup_{c\in
    P} S(c)}\join C(j)|\leq |C(j)|$, and thus $|c\join C(j)|\leq |C(j)|$. Hence 
    this procedure does not increase the size of constraints at nodes after the initial $k$ constraints, immediately resulting in the desired bound of $k$ on the width of $(J,\jtbi)$.

To complete the proof, observe that a join decomposition with the properties outlined above can be found in time at most $|\cspi|^{\bigoh(k)}$: indeed, it suffices to branch over all $k$-element subsets of $C(\cspi)$ and test whether the union of their scopes is functional using, e.g., the result of Cohen et al.~\cite[Corollary 1]{CohenCGM11}. Once we have such a join decomposition, we can solve the instance by invoking Theorem~\ref{the:jt}.
\qed \end{proof}

As a final remark, we note that the class of $k$-constraint rooted CSP
instances naturally includes all instances which contain $k$
constraints that are in conflict (i.e., which cannot all be satisfied
at the same time).  

\subsection{Other Tractable Classes}
\label{sub:otherclasses}
\newcommand{\hkb}{hereditarily $k$-bounded}
Here, we identify some other classes of tractable CSP instances with
bounded joinwidth. First of all, we consider CSP instances such that
introducing their variables in an arbitrary order always results in a
subinstance with polynomially many solutions. In particular, we call a
CSP instance $\cspi$ \emph{\hkb} if for every subset $V'$ of its
variables it holds that $|\SOL(\cspi[V'])|\leq
\noftup(\cspi)^k$. Examples of \hkb{} CSP instances include $k$-Turan
CSPs~\cite[page 12]{CohenCGM11} and CSP instances with  fractional edge covers of weight $k$~\cite{GroheMarx14}.

\begin{PRO}
\label{pro:hereditarybounded}
The class of \hkb{} CSP instances has joinwidth at most $k$ and can be solved in time at most $\bigoh(|\cspi|^k)$.
\end{PRO}

\begin{proof}
Consider an arbitrary linear join decomposition $(J,\jtbi)$. By definition, for each $j\in V(J)$, $|\SOL(\cspi[V(j)])|\leq \noftup(\cspi)^k$. Then $|\proj{\SOL(\cspi[V(j)])}{\overline V(j)}|$ $\leq \noftup(\cspi)^k$, and by Lemma~\ref{lem:jtp-part-sol} we obtain $|C(j)|\leq \noftup(\cspi)^k$, as required.
\qed \end{proof}

Another example of a tractable class of CSP instances that we can
solve using joinwidth are instances where all  constraints interact in a way which forces a unique assignment of the variables. In particular, we say that a CSP $\cspi=\tuple{V,D,C}$ is \emph{unique at depth $k$} if for each constraint $c\in C$ there exists a \emph{fixing set} $C'\subseteq C$ such that $c\in C'$, $|C'|\leq k$, and $|(\join_{c'\in C'}c')|\leq 1$.

\begin{PRO}
\label{pro:uniquedepth}
The class of CSP instances which are unique at depth $k$ has joinwidth at most $k$ and can be solved in time at most $|\cspi|^{\bigoh(k)}$.
\end{PRO}

\lv{
\begin{proof}
First of all, we observe that it is possible to determine whether a CSP instance $\cspi$ is unique at depth $k$ in time at most $|\cspi|^{\bigoh(k)}$. Indeed, it suffices to check whether each constraint $c$ is contained in at least one fixing set---and to do that, we can simply loop over all constraint sets of arity at most $k$ and join them together (in an arbitrary order). If $\cspi$ contains a constraint that does not appear in any fixing set of size at most $k$, then it is not unique at depth $k$; otherwise, for each constraint $c$ we choose an arbitrary fixing set (containing $c$) and denote it as $\mathsf{Fix}(c)$.

Consider a linear join decomposition $\JJJ$ constructed as follows. First, we choose an arbitrary ordering of the constraints and denote these $c_1<\dots<c_m$. We begin our construction by having $\JJJ$ introduce the constraints which occur in $\mathsf{Fix}(c_1)$ (in an arbitrary order). Then, for each $2\leq i\leq m$, we introduce those constraints in $\mathsf{Fix}(c_i)$ which are not yet introduced in $\JJJ$ (also in an arbitrary order).

It remains to argue that $\JJJ$ has joinwidth at most $k$. Consider the node $j_1$ whose child is the last constraint occurring in $\mathsf{Fix}(c_1)$ (i.e., all constraints in $\mathsf{Fix}(c_1)$ occur as leaves which are descendants of $j_1$, and all constraints not in $\mathsf{Fix}(c_1)$ occur as leaves which are not descendants of $j_1$). By definition, $|\join_{c'\in \mathsf{Fix}(c_1)}c'|\leq 1$ and hence by Lemma~\ref{lem:jtp-part-sol} we obtain $|C(j_1)|\leq 1$. Moreover, since there are at most $k$ leaves below $j_1$, for each descendant $j_1'$ of $j_1$ we obtain that $|C(j_1')|<(\noftup(\cspi))^k$, as required. Now, consider an arbitrary $2\leq i\leq m$, and as before let $j_i$ be the node whose child is the last constraint occurring in $\mathsf{Fix}(c_i)$. Once again, $|C(j_i)|\leq 1$ by definition. Moreover, since there are at most $k$ leaves which occur ``between'' $j_i$ and $j_{i-1}$, $\JJJ$ contains at most $k-1$ non-leaf nodes between $j_i$ and $j_{i-1}$. Since $|C(j_{i-1})|\leq 1$, it follows that every node $j_i'$ between $j_i$ and $j_{i-1}$ also has $|C(j_1')|<(\noftup(\cspi))^k$. Hence $\JJJ$ indeed has joinwidth at most $k$, and the proposition follows by Theorem~\ref{the:jt}.
\qed \end{proof}
}

\section{Solving Bounded-Width Instances}\label{sec:solving}
This section investigates the tractability of CSP instances whose joinwidth is bounded by a fixed constant $\omega$. In particular, one can investigate two notions of tractability. The first one is the classical notion of \emph{polynomial-time tractability}, which asks for an algorithm of the form $|\cspi|^{\bigoh(1)}$. 
In this setting, the complexity of CSP instances of bounded joinwidth remains an important open problem. Note that the \NP{}-hardness
of the $\omega$-\COMPJT{} problem established in Theorem~\ref{the:comp-jt-np}
does not exclude polynomial-time tractability for CSP instances of
bounded joinwidth. For instance, tractability could still be
obtained with a suitable approximation algorithm for computing
join decompositions (as it is the case for fractional
hypertreewidth~\cite{Marx10b}) or by using an algorithm that does not require a join decomposition
of bounded width as input.

The second notion of tractability we consider is called \emph{fixed\hy parameter tractability} and asks for an algorithm of the form $f(k)\cdot |\cspi|^{\bigoh(1)}$, where $k$ is a numerical parameter capturing a certain natural measure of $\cspi$. Prominently, Marx investigated the fixed\hy parameter tractability of CSP and showed that CSP instances whose hypergraphs have bounded \emph{submodular width}~\cite{Marx13} are fixed\hy parameter tractable when $k$ is the number of variables. Moreover, Marx showed that submodular width is the most general structural property \emph{among those measured purely on hypergraphs} with this property.

Here, we obtain two single-exponential fixed\hy parameter algorithms
for CSP instances of bounded joinwidth (i.e., algorithms with a
running time of
$2^{\bigoh(k)}\cdot |\cspi|^{\bigoh(1)}$): one where $k$ is the number
of variables, and the other where $k$ is the number of
constraints. Since there exist classes of instances of bounded
joinwidth and unbounded submodular width (see
Observation~\ref{obs:fhtw}), this expands the frontiers of
(fixed-parameter) tractability for CSP. 


\newcommand{\dnm}{\alpha}
\newcommand{\mundef}{\mathbf{u}}
\newcommand{\exceed}{\infty}
\newcommand{\completecon}{\emptyset}

\paragraph{Parameterization by Number of Constraints.} To solve the case where $k$ is the number of constraints, our primary aim is to obtain a join decomposition of width at most $\omega$, i.e., solve the $\omega$-\COMPJT\ problem defined in Section~\ref{sec:jw}. Indeed, once that is done we can solve the instance by Theorem~\ref{the:jt}. 

%
\begin{THE}
\label{the:compjt-con}
  $\omega$-\COMPJT{} can be solved in time $\bigoh(4^{|C|}+2^{|C|}|\cspi|^{2\omega+1})$ and is hence
  fixed\hy parameter tractable parameterized by $|C|$, for a CSP instance
  $\cspi=\tuple{V,D,C}$.
\end{THE}
\lv{\begin{proof}
  Let $\cspi=\tuple{V,D,C}$ be the given CSP instance. For a subset
  $C' \subseteq C$, we denote by $S(C')$ the set of all variables that
  are both in the scope of a constraint in $C'$ and in the scope of a
  constraint in $C \setminus C'$.
  The algorithm uses dynamic programming to compute the mapping $\dnm$
  that for every non-empty subset $C'$ of $C$ either:
  \begin{itemize}
  \item is equal to $\proj{\SOL(\cspi[V(C')])}{S(C')}$, if there is a
    partial join decomposition for $\cspi$ that covers the constraints in $C'$ of width at most $\omega$, and
  \item otherwise is equal to $\exceed$.
  \end{itemize}
  Note that $\cspi$ has a join decomposition of width at most $\omega$ if and
  only if $\dnm(C)\neq \exceed$.

  To compute $\dnm$ we use the observation that a partial join decomposition
  covering a set $C'$ of constraints is either a leaf (if $|C'|=1$) or is
  obtained as the ``join'' of a partial join decomposition covering $C_0$ and a partial
  join decomposition covering $C' \setminus C_0$ for some non-empty $C_0 \subset C'$. This immediately implies that $\dnm$ satisfies the
  following recurrence relation for every non-empty subset $C'$ of $C$:
  \begin{itemize}
  \item[(R1)] if $C'=\{c\}$ for some $c \in C$, then $\dnm(C')=\PR(\proj{c}{S(C')})$,
  \item[(R2)] if $|C'|>1$ and there is a non-empty subset $C_0 \subset C'$ such that:
    \begin{itemize}
    \item[(C1)] $\dnm(C_0)\neq \exceed$ and $\dnm(C' \setminus C_0)\neq \exceed$, and
    \item[(C2)] $|\PR((\proj{\dnm(C_0)\join \dnm(C'\setminus C_0))}{S(C')})| \leq
      \noftup(\cspi)^\omega$,
    \end{itemize}
    then $\dnm(C')=\PR(\proj{(\dnm(C_0)\join \dnm(C\setminus
      C_0))}{S(C')})$,
  \item[(R3)] otherwise, i.e., if $|C'|>1$ and (R2) does not apply,
    then $\dnm(C')=\exceed$.
  \end{itemize}
  Since $\cspi$ has a join decomposition of width at most $\omega$ if and
  only if $\dnm(C)\neq \exceed$ and in this case a join decomposition for
  $\cspi$ can be easily constructed by following the recurrence
  relation starting from $\dnm(C)$, it only remains to show how to compute
  $\dnm$ in the required running time, which we do as follows.
  
  Initially, we
  set $\dnm(\{c\})=\PR(\proj{c}{S(C')})$ 
  for every $c \in C$ and $\dnm(C')=\exceed$
  for every $C'$ with $\emptyset \neq C' \subseteq C$ and $|C'|>1$.
  We then enumerate all subsets $C'$ with $C' \subseteq C$ and
  $|C'|>1$ in order of increasing size and check whether there is a subset
  $C_0$ of $C'$ satisfying the conditions stated in (R2). If so we set
  $\dnm(C')=\PR((\proj{\dnm(C_0)\join \dnm(C\setminus C_0))}{S(C')})$
  and otherwise we set $\dnm(C')=\exceed$. Clearly, the time required
  to initialize $\dnm$ is at most $\bigoh(2^{|C|}+|\cspi|)$. 
  Furthermore, the
  time required to check, whether a subset $C'$ satisfies the
  conditions stated in (R2) is at most
  $\bigoh(2^{|C'|}+|\cspi|^{2\omega+1})=\bigoh(2^{|C|}+|\cspi|^{2\omega+1})$, which
  can be obtained as follows. First note that because of
  Lemma~\ref{lem:jtp-part-sol}, it holds that $\PR(\proj{(\dnm(C_0)\join
    \dnm(C'\setminus C_0))}{S(C')})$ is equal to $\PR((\proj{\dnm(C_0')\join
    \dnm(C'\setminus C_0'))}{S(C')})$ for any two subsets $C_0$ and
  $C_0'$ satisfying (C1) and hence once we found a subset $C_0$
  satisfying (C1) and have computed $\PR((\proj{\dnm(C_0)\join
    \dnm(C'\setminus C_0))}{S(C')})$, we can determine whether (R2) or
  (R3) applies for $C'$. This implies that we have to compute
  $\PR((\proj{\dnm(C_0)\join \dnm(C'\setminus C_0))}{S(C')})$ for at
  most one subset $C_0$, which explains why $|\cspi|^{2\omega+1}$ only
  appears additively in the running time. Moreover, the term $2^{|C|}$
  is required for the enumeration of all subsets $C_0$ of $C'$.
  Since we have to enumerate
  all subsets $C'$ of $C$, we hence obtain
  $\bigoh(2^{|C|}(2^{|C|}+|\cspi|^{2\omega+1}))=\bigoh(4^{|C|}+2^{|C|}|\cspi|^{2\omega+1})$
  as the total running time for computing $\dnm$.
\qed \end{proof}}

From Theorem~\ref{the:compjt-con} and Theorem~\ref{the:jt} we immediately obtain:
\begin{COR}
  \label{cor:fptconst}
  A CSP instance $\cspi$ with $k$ constraints and joinwidth at most
  $\omega$ can be solved in time $2^{\bigoh(k)}\cdot |\cspi|^{\bigoh(\omega)}$.
\end{COR}

\paragraph{Parameterization by Number of Variables.} Note that
Corollary~\ref{cor:fptconst} immediately establishes fixed\hy parameter
tractability for the problem when $k$ is the number of variables
(instead of the number of constraints), because one can assume that $|C|\leq 2^{|V|}$ for
every CSP instance $\cspi=(V,D,C)$. However, the resulting algorithm
would be double-exponential in $|V|$. The following theorem shows that
this can be avoided by designing a dedicated algorithm for \CSP{}
parameterized by the number of variables. The main idea behind both
algorithms is dynamic programming, however, in contrast to the
algorithm for $|C|$, the table entries for the fpt-algorithm for
$|V|$ correspond to subsets of $V$ instead of subsets of
$C$. Interestingly, the fpt-algorithm for $|V|$ does not explicitly
construct a join decomposition, but only implicitly relies on the existence of one.

\begin{THE}
  \label{the:fptvar}
  A CSP instance $\cspi$ with $k$ variables and joinwidth at most
  $\omega$ can be solved in time $2^{\bigoh(k)}\cdot |\cspi|^{\bigoh(\omega)}$.
\end{THE}
\lv{
\begin{proof}
  Let $\cspi=\tuple{V,D,C}$ be the given CSP instance with
  $\jw(\cspi)\leq \omega$.
    For a subset $V' \subseteq V$, we denote by $S(V')$ the set of
  variables in $V' \cap \SB v \SM c \in C \land S(c)\setminus
  V'\neq \emptyset \land v \in S(c)\SE$ (i.e., $S(V')$ contains variables in $V'$ which occur in the scope of a constraint that also has variables outside of $V'$).
  The algorithm uses dynamic programming to compute a mapping $\dnm$
  which maps non-empty subsets $V'$ of $V$ to either the 
  constraint $\proj{\SOL(\cspi[V'])}{S(V')}$, or to $\exceed$.
  $\dnm$ is defined using the following recurrence:
  \begin{itemize}
  \item[(R1)] if $V'=S(c)$ for some $c \in C$, then $\dnm(V')=\PR(\proj{c}{S(V')})$,
  \item[(R2)] if (R1) does not apply and there are subsets $V_0$ and $V_1$ with $\emptyset \neq
    V_0,V_1 \subset V'$ and $V_0\cup V_1=V'$ such that:
    \begin{itemize}
    \item[(C1)] $\dnm(V_0)\neq \exceed$ and $\dnm(V_1)\neq \exceed$, and
    \item[(C2)] $|\PR(\proj{\dnm(V_0)\join \dnm(V_1)}{S(V')})| \leq
      \noftup(\cspi)^\omega$,
    \end{itemize}
    then $\dnm(V')=\PR(\proj{\dnm(V_0)\join \dnm(V_1)}{S(V')})$,
  \item[(R3)] otherwise, i.e., if neither (R1) nor (R2) applies,
    then $\dnm(V')=\exceed$.
  \end{itemize}
  Using arguments very similar to the ones used in the proof of
  Lemma~\ref{lem:jtp-part-sol} it is now easy to show that
  if $\dnm(V')\neq \exceed$, then
  $\dnm(V')=\proj{\SOL(\cspi[V'])}{S(V')}$ for every subset $V'
  \subseteq V$.
  
  It follows that if $\dnm(V)\neq \exceed$, then $\cspi$ has a
  solution if and only if $\dnm(V)$ is not equal to the empty
  constraint. Hence an algorithm that computes $\dnm$ can be used to
  solve every CSP instance for which $\dnm(V)\neq \exceed$. In order to
  show that we can solve a CSP instance $\cspi$ of joinwidth at most $\omega$, it is hence sufficient to show that
  $\dnm(V)\neq \exceed$. 
  
  To this end, consider a
  join decomposition $\JJJ=(J,\jtbi)$ for $\cspi$ of width at most
  $\omega$. Using a bottom-up induction on $J$, we will show that
  $\dnm(V(j))\neq \exceed$ and moreover $C(j)=\dnm(V(j))$ for every $j
  \in V(j)$. Since $\dnm(V)=C(r)\neq \exceed$ for the root $r$ of $J$,
  this then implies the desired property.
  $C(j)=\dnm(V(j))$ clearly holds for every leaf $l$ of $J$, because $V(l)$ satisfies (R1) and by 
  Lemma~\ref{lem:jtp-part-sol} we have
  $C(l)=\proj{\SOL(\cspi[V(l)])}{S(l)}=\proj{\SOL(\cspi[V(l)])}{S(V(l))}$.
  Now, consider an inner node $j$ of $J$ with children $j_1$ and
  $j_2$. Then setting $V'$ to $V(j)$, $V_0$ to $V(j_1)$, and $V_1$ to $V(j_2)$
  satisfies (C1) (of (R2)) because of the induction
  hypothesis. Moreover, (C2) is also satisfied because
  $\PR(\proj{\dnm(V_0)\join
    \dnm(V_1)}{S(V')})=\proj{\SOL(\cspi[V'])}{S(V')}=\proj{\SOL(\cspi[V(j)])}{S(j)}=C(j)$
  and $|C(j)|\leq \noftup(\cspi)^\omega$. Hence an algorithm that
  computes $\dnm$ can be used to solve every CSP instance of bounded
  joinwidth, and it only remains to show how to compute $\dnm$.
  
  Initially, we set $\dnm(V')=\PR(c)$ for every $V'\subseteq V$ such
  that there is constraint $c \in C$ with $V'= S(c)$ and
  $\dnm(V')=\exceed$ for every $V' \subseteq V$ for which this is not the case.
  We then enumerate all subsets $V'$ with $\emptyset \neq V' \subseteq V$ 
  in order of increasing sizes and if $\dnm(V')=\exceed$, we check whether there are subsets
  $V_0$ and $V_1$ of $V'$ satisfying the conditions stated in (R2). If so we set
  $\dnm(V')=\PR(\proj{\dnm(V_0)\join \dnm(V_1)}{S(V')})$
  and otherwise we set $\dnm(V')=\exceed$. Clearly, the time required
  to initialize $\dnm$ is at most $\bigoh(2^{|V|}+|C|)$. 
  Furthermore, the
  time required to check whether a subset $V'$ satisfies the
  conditions stated in (R2) is at most
  $\bigoh(2^{|V|}2^{|V|}+|\cspi|^{2\omega+1})=\bigoh(4^{|V|}+|\cspi|^{2\omega+1})$---indeed, this can be carried out as follows. To begin, we observe that 
  $\PR(\proj{\dnm(V_0)\join
    \dnm(V_1)}{S(V')})$ is equal to $\PR(\proj{\dnm(V_0')\join
    \dnm(V_1')}{S(V')})$ for any two pairs $(V_0,V_1)$ and
  $(V_0',V_1')$ of subsets satisfying (C1), and hence after we find a
  pair $(V_0,V_1)$ of subsets
  satisfying (C1) and compute $\PR(\proj{\dnm(V_0)\join
    \dnm(V_1)}{S(V')})$, we can determine whether (R2) or
  (R3) applies for $V'$. This implies that we have to compute
  $\PR(\proj{\dnm(V_0)\join \dnm(V_1)}{S(V')})$ for at
  most one pair $(V_0,V_1)$ of subsets, which explains why $|\cspi|^{2\omega+1}$ only
  appears additively in the running time. Moreover, the term $4^{|V|}$
  is required for the enumeration of all pairs $(V_0,V_1)$ for $V'$. Since we have to enumerate
  all subsets $V'$ of $V$, we obtain
  $\bigoh(2^{|V|}(4^{|V|}+|\cspi|^{2\omega+1}))=\bigoh(6^{|V|}+2^{|V|}|\cspi|^{2\omega+1})$,
  as the total running time for computing $\dnm$.
\qed \end{proof}
}

\section{Beyond Join Decompositions}\label{sec:beyond}

Due to their natural and ``mathematically clean'' definition, one might be tempted to
think that join decompositions capture all the algorithmic power offered by join and
projection operations. It turns out that this is not the case,
i.e., we show that if one is allowed to
use join and projections in an arbitrary manner (instead of the more
natural but also more restrictive way in which they are used within join
decompositions) one can solve CSP instances that are out-of-reach even for
join decompositions. This is interesting as it points towards the
possibility of potentially more powerful parameters based on join and
projections than joinwidth. 
\begin{THE}
\label{the:beyond}
  For every $\omega$, there exists a CSP instance $\cspi_\omega$ that can
  be solved in time $\bigoh(|\cspi|^4)$ using only join and projection operations but 
  $\jw(\cspi_\omega)\geq \omega$.
\end{THE}
\lv{
\begin{proof}
  Let $n=\omega*16$ and let $\cspi_\omega=\tuple{V,D,C}$ be the CSP instance with variables
  $x_1,\dotsc,x_n$, domain $\{1,\dotsc,n\}$, and the following
  constraints:
  \begin{itemize}
  \item A set $C_<=\SB c_1,\dotsc,c_{n-1}\SE$ of binary constraints $c_i$
    with scope $(x_i,x_{i+1})$ containing all tuples $t \in n \times
    n$ such that $t[x_i] < t[x_{i+1}]$,
  \item A set $C_C=\SB c_{l,m} \SM l < m-1 \land 1 \leq l,m \leq n
    \SE$ of binary and complete constraints $c_{l,m}$ with scope $(x_l,x_m)$.
  \end{itemize}
  We start by showing that $\jw(\cspi_\omega)\geq \omega$. Let
  $(J,\jtbi)$ be any join decomposition for $\cspi_\omega$. Let $j$ be a node of
  $J$ such that $\lceil 1/4n\rceil \leq |V(j)| \leq
  \lfloor 1/2n \rfloor$, which
  exists due to the following well-known argument.
  
  First we show that there is a node $j'$ of $J$ such that $|V(j')| \geq \lfloor
  1/2n\rfloor$ but $|V(j'')| \leq \lfloor 1/2n\rfloor$
  for every child $j''$ of $j'$ in $J$. Namely, $j'$ can be found by
  going down the tree $J$ starting from the root and choosing a child
  $j''$ of $j'$ with $|V(j'')| \geq \lfloor 1/2n\rfloor$; as long as
  such a child $j''$ exists. Then letting $j$ be the child
  of $j'$ maximizing $|V(j)|$ implies that
  $\lceil 1/4n \rceil \leq |V(j)|\leq \lfloor 1/2n\rfloor$,
  as required.

  Note that
  because $\cspi_\omega$ has a constraint with scope $(x_l,x_m)$ for
  every $l<m$, it holds that $\overline{V}(j)=V$ for the node $j \in V(J)$;
  this is because there is at least one variable in $V \setminus V(j)$
  and this variable occurs in a binary constraint with every other variable.
  Hence by Lemma~\ref{lem:jtp-part-sol}, it
  holds that $C(j)=\SOL(\cspi_\omega[V(j)])$.
  
  We claim that $C(j)$ contains at least $(n/2)^{n/4}$
  tuples. Towards showing the claim, let $\cspi_\omega(i)$ be any
  sub-instance of $\cspi_\omega$ induced by any subset of $V$ of size
  $i$. It is easy to see that the domain of any variable in
  $\cspi_\omega(i)$ has size at least $n-i$ and hence
  $\SOL(\cspi_\omega)$ has at least $(n-i)^i$ tuples. Hence, because
  $\lceil 1/4n \rceil|V(j)|\leq \lfloor 1/2n\rfloor$, we obtain that
  $C(j)$ has at least $\min \{(3/4n)^{1/4n},(1/2n)^{1/2n}\}\geq
  (n/2)^{n/4}$ tuples, as required.

  Finally, since $\noftup(\cspi_\omega)=n^2$, we obtain the following for the
  joinwidth of $\cspi_\omega$.
  \begin{eqnarray*}
    \jw(\cspi_\omega) & = & \log_{n^2}((n/2)^{n/4})\\
                      & = & n/2\log_{n^2}n/4 \\
                      & = & n/2(\log_{n^2}n-\log_{n^2}4) \\
                      & = & n/2(1/2-\frac{\log_44}{\log_4n^2}) \\
                      & = & n/2(1/2-\frac{1}{\log_4n^2})\\
                      & \geq & n/2(1/2-\frac{1}{4}) \\
                      & = & n/8 \\
                      & = & 2\omega
  \end{eqnarray*}
  The inequality in the above sequence follows since $n^2\geq 16^2=4^4$.

  It remains to show that we can solve $\cspi_\omega$ in time
  $\bigoh(|\cspi_\omega|^4)$ using only join and projection operations.

  Towards showing this we now define various auxiliary constraints that can be
  obtained efficiently using only join and projection
  operations from the constraints in $C$.
  
  Namely, for every $i$ with $1 < i \leq n$ let $b_i^\uparrow$  be the
  constraint defined iteratively as follows. Set $b_2^\uparrow$ to be
  the constraint $\proj{c_1}{\{x_2\}}$ and for every $i>2$, set
  $b_i^\uparrow$ to be the constraint $\proj{b_{i-1}^\uparrow \join
    c_{i-1}}{\{x_i\}}$. Note that $S(b_i^\uparrow)=\{x_i\}$ and
  $R(b_i^\uparrow)$ contains all tuples $t$ with $t[x_i] \in
  \{i,\dotsc,n\}$. Moreover, it follows immediately from their
  defintion that the constraints $b_2^\uparrow,\dotsc,b_n^\uparrow$ can be computed in time at most
  $\bigoh(n^4)=\bigoh(|\cspi_\omega|^4)$ using only join and projection operations.

  Similarily, let $b_i^\downarrow$ for every $i$ with $1 \leq i < n$ be the
  constraint defined recursively as follows. Let $b_{n-1}^\downarrow$ be
  the constraint $\proj{c_{n-1}}{\{x_{n-1}\}}$ and for every $i<n-1$, we set
  $b_i^\downarrow$ to be the constraint $\proj{b_{i+1}^\downarrow \join
    c_i}{\{x_i\}}$. Note that $S(b_i^\downarrow)=\{x_i\}$ and
  $R(b_i^\downarrow)$ contains all tuples $t$ with $t[x_i] \in
  \{1,\dotsc,i\}$. Moreover, it follows immediately from their
  definition that the constraints
  $b_1^\downarrow,\dotsc,b_{n-1}^\downarrow$ can be computed in time at most
  $\bigoh(n^4)=\bigoh(|\cspi_\omega|^4)$ using only join and projection operations.

  Note that $b_i^\uparrow \join b_i^\downarrow$ for every $i$ with $1
  < i < n$ has scope $\{x_i\}$ and contains only one tuple, i.e., the
  tuples $t$ with $t[x_i]=i$. Moreover, $b_1^\downarrow$ and
  $b_{n}^\uparrow$ also contain only one tuple, i.e., the tuples $t$
  with $t[x_1]=1$ and $t[x_n]=n$, respectively. To use this
  observation let $b_1$ be the constraint $b_1^\downarrow$, $b_n$ be
  the constraint $b_n^\uparrow$ and for every $i$ with $1 < i < n$
  let $b_i$ be the constraint $b_i^\uparrow \join
  b_i^\downarrow$. Then $b_i$ has scope $\{x_i\}$ and $R(b_i)$
  contains only one tuple, i.e., the tuple $t$ with
  $t[x_i]=i$. Moreover, it follows immediately from the definition of
  the constraints $b_1,\dotsc, b_n$ that they can be computed from the
  constraints $b_i^{\uparrow},b_i^\downarrow$ in time at most
  $\bigoh(n^2)=\bigoh(|\cspi_\omega|^2)$ using only join
  operations.
  
  Note that at this stage, we have already identified the unique
  solution of $\cspi_\omega$; the one that sets every variable $x_i$
  to the value $i$. However, since we have not yet even considered all
  constraints of $\cspi_\omega$, we cannot yet bet sure that what we computed is actually a solution
  of $\cspi$. To circumvent this caveat, we now compute the
  constraint $b=b_1\join \dotsb \join b_n$ in time $\bigoh(n)$; note
  that $b$ has scope $V$ and $R(b)$ only contains the tuple $t$ with
  $t[x_i]=i$ for every $i$ with $1 \leq i \leq n$. Finally, we join
  every constraint $c \in C$ with $b$ one-by-one as follows. Assume
  that $C=\{c^1,\dotsc,c^{|C|}\}$. Then for every $i$ with $1 \leq i
  \leq |C|$, we now compute the constraint $a_i$ iteratively by
  setting: $a_1$ to be the constraint $b \join c^1$ and for every
  $i>1$, $a_i$ is the contraint $a_{i-1}\join c^i$. Now we can be sure
  that the constraint $a_{|C|}$ contains all solutions (and only the
  solutions) of $\cspi$. Since $a_i$ is equal to $b$ for every $i$ it
  also follows that we can compute the constraints
  $a_1,\dotsc,a_{|C|}$ in time $\bigoh(|C|)=\bigoh(|\cspi_\omega|)$
  using only join operations.
\qed \end{proof}
}

\section{Conclusions and Outlook}\label{sec:conclusion}
The main contribution of our paper is the introduction of the
notion of a join decomposition and the associated parameter joinwidth
(Definitions~\ref{def:jointree} and \ref{def:width}). These notions are
natural as they are entirely based on fundamental operations of
relational algebra: joins, projections, and pruning (which can equivalently be
stated in terms of semijoins). It is also worth noting that our algorithms seamlessly extend
to settings where each variable has its own domain (this can be modeled, e.g., by unary constraints).

\lv{
We establish several structural and complexity results that put our
new notions into context. In particular, we show that:
\begin{enumerate}
\item bounded joinwidth captures and properly
  contains several known restrictions that render CSP tractable
  (Theorem~\ref{the:jtw-fhtw}, Propositions~\ref{pro:rootset}, \ref{pro:hereditarybounded} and~\ref{pro:uniquedepth}). 
\item\label{r:poly} CSP instances of bounded joinwidth can be
  solved in polynomial time assuming the corresponding join decomposition is
  provided with the input (Theorem~\ref{the:jt});
\item\label{r:nph} finding a join decomposition of optimal width
  is NP-hard, already for a constant upper bound on the width
  (Theorem~\ref{the:comp-jt-np}), mirroring the situation surrounding
  fractional hypertree width \cite{FischlGottlobPichler18};
\item CSP
  instances of bounded joinwidth can be solved by a single-exponential 
  fixed-parameter algorithm parameterized either by the number of constraints
  (Corollary~\ref{cor:fptconst}) or the number of variables
  (Theorem~\ref{the:fptvar});
\item\label{r:submodular} there are instances of bounded
  joinwidth but unbounded submodular (or adaptive) width
  (Observation~\ref{obs:fhtw}); bounded submodular width is the most
  general hypergraph restriction that allows for fixed\hy parameter
  tractability of CSP under the Exponential Time Hypothesis~\cite{Marx13};
\item \label{r:beyond} using joins and projections,
  one can even solve instances of unbounded joinwidth
  (Theorem~\ref{the:beyond}).
\end{enumerate}
} 

\lv{
Our results give rise to several interesting directions 
for future work. We believe that result~(\ref{r:poly}) can be generalized to other
problems, such as \#CSP
or the FAQ-Problem \cite{KhamisNgoRudra16}.
Result (\ref{r:nph}) gives rise to the question of whether there
exists a polynomial-time approximation algorithm for computing
join decompositions of suboptimal joinwidth, similar to Marx's algorithm for
fractional hypertree-width \cite{Marx10b}. One can also try to develop
an efficient heuristic approach for computing the exact joinwidth, possibly similar to the SMT-encoding for fractional hypertree width as recently proposed by Fichte et
al.~\cite{FichteHecherLodhaSzeider18}.

Result (\ref{r:submodular}) shows that submodular width is not more
general than joinwidth. We conjecture that also the converse
direction holds, i.e., that the two parameters are actually
incomparable. Motivated by result (\ref{r:beyond}), one could try to define a natural
parameter that captures the full generality of join and
projection operations, or to at least define a parameter that is more
general than join decompositions without sacrificing the simplicity of the
definition.
}

\bibliographystyle{plainurl}
\bibliography{literature}
 
\end{document}